%% file: main.tex
\documentclass{article}

\usepackage{amssymb}
\usepackage[linesnumbered, lined, boxed]{algorithm2e}
\usepackage{amsthm}
\usepackage{amsmath}
\usepackage{hyperref}
\usepackage{cleveref}
\usepackage{comment}
\usepackage{xcolor}
\usepackage{fullpage}
\usepackage{tikz}
\usepackage{pgfplots}
\usepackage{ifoddpage}
\usepackage{float}
\usepackage{microtype}
\usepackage[utf8]{inputenc}
\usepackage[T1]{fontenc}
\pgfplotsset{compat=1.5}

\newtheorem{theorem}{Theorem}
\newtheorem{definition}[theorem]{Definition}

\newtheorem{lemma}[theorem]{Lemma}
\newtheorem{claim}[theorem]{Claim}

\numberwithin{theorem}{section}

\newcommand{\kCenter}{\mathsf{center}}
\newcommand{\totalCover}{\mathsf{totalCover}}
\newcommand{\grid}{\mathsf{grid}}
\newcommand{\cover}{\mathsf{cover}}
\newcommand{\Lap}{\mbox{Lap}}
\newcommand{\snap}[2]{\lfloor #1 \rfloor^{(#2)}}
\newcommand{\pluseq}{\mathrel{{+}{=}}}

\newcommand{\E}{\mathop{\mathbb{E}}}
\newcommand{\ex}[2]{{\ifx&#1& \mathbb{E} \else \underset{#1}{\mathbb{E}} \fi \left[#2\right]}}
\newcommand{\pr}[2]{{\ifx&#1& \mathbb{P} \else \underset{#1}{\mathbb{P}} \fi \left[#2\right]}}

\newcommand{\R}{\mathbb{R}}
\newcommand{\eps}{\varepsilon}
\newcommand{\poly}{\mathrm{poly}}
\DeclareMathOperator*{\argmin}{arg\,min}

\newcommand{\OPT}{\texttt{OPT}}

\newcommand{\C}{\mathcal{C}}

\newcommand{\U}{\mathcal{U}}

\newcommand{\Z}{\mathcal{Z}}

\allowdisplaybreaks

\title{Differentially private $k$-means clustering via exponential mechanism and max cover}

\author{Anamay Chaturvedi\thanks{Khoury College of Computer Sciences, Northeastern University \dotfill \texttt{chaturvedi.a@northeastern.edu}} \and Huy L\^{e} Nguy\~{\^{e}}n \thanks{Khoury College of Computer Sciences, Northeastern University \dotfill \texttt{hlnguyen@cs.princeton.edu}}\and Eric Xu\thanks{Khoury College of Computer Sciences, Northeastern University \dotfill \texttt{xu.er@northeastern.edu}}}

\begin{document}
	\maketitle
	
	\begin{abstract}
		We introduce a new $(\epsilon_p, \delta_p)$-differentially private algorithm for the $k$-means clustering problem. Given a dataset in Euclidean space, the $k$-means clustering problem requires one to find $k$ points in that space such that the sum of squares of Euclidean distances between each data point and its closest respective point among the $k$ returned is minimised. Although there exist privacy-preserving methods with good theoretical guarantees to solve this problem (\cite{balcan2017differentially, KS18}), in practice it is seen that it is the additive error which dictates the practical performance of these methods. By reducing the problem to a sequence of instances of maximum coverage on a grid, we are able to derive a new method that achieves lower additive error then previous works. For input datasets with cardinality $n$ and diameter $\Delta$, our algorithm has an $O(\Delta^2 (k \log^2 n  \log(1/\delta_p)/\epsilon_p + k\sqrt{d \log(1/\delta_p)}/\epsilon_p))$ additive error whilst maintaining constant multiplicative error. We conclude with some experiments and find an improvement over previously implemented work for this problem.
	\end{abstract}

	\input{introduction.tex}

	\input{preliminaries.tex}
	
	\input{lower_bounds}
    
    \input{algorithm.tex}

	\input{utility.tex}

    \input{privacy.tex}
    
	\input{experiments.tex}

	\bibliography{biblio}
	\bibliographystyle{plain}
    
\end{document}

%% file: introduction.tex
\section{Introduction}	
	Clustering is a well-studied problem in theoretical computer science. The objective can vary between identifying good cluster sets, as one might desire in unsupervised learning, or good cluster centers, which may be framed as more of an optimization problem. One relatively general variant of this problem is to find a fixed number of centers $k$ for a given dataset $D$ of size $n$ such that the sum of distances of each point to the closest center is minimized. When the ambient space is Euclidean and the distance is the square of the Euclidean metric this is known as the $k$-means problem. Although solving the $k$-means problem is NP hard \cite{aloise2009np, dasgupta2008hardness, mahajan2009planar}, practical algorithmic solutions with good approximation guarantees are well-known \cite{lloyd1982least, ostrovsky2013effectiveness, kanungo2004local, ahmadian2019better}. 
	
	When algorithms handle sensitive information (for example location data), an important requirement that they might be expected to fulfill is that of being \emph{differentially private} \cite{dwork2006calibrating}. Differential privacy provides a framework for capturing the loss in privacy that occurs when sensitive data is processed. This framework encompasses different models that vary depending on who the trusted parties are, what is considered sensitive data, and how loss in privacy is measured. In this work we are interested in the centralized model of differential privacy, where we assume that the algorithm whose privacy loss we want to bound is executed by a trusted curator with access to many agents' private information. This trusted curator publicly reveals the output obtained at the end of their computation, which is when all privacy loss occurs.
	
	In the theoretical study of the $k$-means problem, reducing the worst-case multiplicative approximation factor has been the focus of a major line of work \cite{kanungo2004local, ahmadian2019better}. However, it is important to note that Lloyd's algorithm, which has a tight sub-optimal multiplicative guarantee of $O(\log k)$, works quite well in practice. This behaviour can be understood by showing (as in \cite{aggarwal2009adaptive}) that Lloyd's finds a solution with constant multiplicative error with constant probability, or that for a general class of datasets satisfying a certain separability condition \cite{ostrovsky2013effectiveness} the multiplicative error again has a strong $O(1)$ bound.
	
	In contrast, when the algorithm is required to be $(\epsilon_p, \delta_p)$-differentially private, no pure multiplicative approximation is attainable and additive error is necessary. Differential privacy forces a lower bound on the clustering cost - morally, for an easy clustering instance a perfect solution would reveal too much information about the sensitive dataset. This is formalised for the closely related discrete $k$-medians\footnote{The discrete $k$-medians problem has an identical formulation to the $k$-means except that the distance function is a metric, not squared, and the centers come from a public finite set and not the whole ambient space.} problem in theorem 4.4 of \cite{gupta2010differentially} which shows that there is a family of instances whose optimal clustering cost is $0$ but any differentially private algorithm must incur an $\Omega(\Delta^2 k (\log n/k)/\epsilon_p)$ expected cost. In practice, for many datasets it is seen that although the non-private clustering cost naturally decreases as the number of centers $k$ increases, the costs incurred by differentially private algorithms quickly plateau (as in the experiments of \cite{balcan2017differentially}), suggesting that they have reached their limit in the additive error. Given this fundamental barrier, a major question is:
	
	\paragraph{Question:} Is it possible to obtain a finite approximation with additive error nearly linear in $k$?

    \subsection{Related work}
    The work \cite{gupta2010differentially} initiated the study of differentially private clustering algorithms and gave a private algorithm for the $k$-means problem with constant factor multiplicative approximation and $\tilde{O}(k^{1.5})$ additive error in fixed dimensions. Motivated by applications in high dimensions, the work~\cite{balcan2017differentially} developed a different method that scales to high dimensions but with an $O(\log^3 n)$ multiplicative approximation and $\tilde{O}(k^{1.5})$ additive error. Subsequently, the work~\cite{KS18} gave an algorithm with $O(1/\gamma)$ approximation and $\tilde{O}(k^{1.5}+k^{1+\gamma}d^{0.5+\gamma})$ additive error with one key idea (among others) being the application of locality sensitive hashing (LSH). The trade-off between the multiplicative and additive errors is in part a consequence of the trade-off between the distance approximation quality of LSH and the number of hash functions. Recently, the work~\cite{JNN} improved the additive error to $\tilde{O}(k^{1+\gamma}d^{0.5+\gamma})$ but the multiplicative error remains $O(1/\gamma)$. In this work, we simultaneously improve upon the minimum additive error in this trade-off to nearly linear in $k$ and eliminate the resulting blow-up in the multiplicative factor using a very different approach.
    
    \begin{table}
    \centering
    \renewcommand{\arraystretch}{1.5}
    \begin{tabular}{l|c|c}
    \textbf{Reference} & \textbf{Multiplicative Error} & \textbf{Additive Error}\\ \hline
    Balcan et al. \cite{balcan2017differentially}  & $O\left(\log^3n\right)$ & $\tilde{O}\left(k^2+d \right)$ \\ \hline
    Kaplan and Stemmer \cite{KS18} & $O(1/\gamma)$ & $ \tilde{O}\big(k^{1.5} + d^{0.5 + \gamma} k^{1 + \gamma} \big)$ \\ \hline 
    Jones, Nguyen and Nguyen \cite{JNN} & $O(1/\gamma)$ & $ \tilde{O}\big(k + d^{0.5 + \gamma} k^{1 + \gamma} \big)$ \\ \hline 
    \textbf{Ours}& $O(1)$ & $\tilde{O}(k \sqrt{d})$ \\ 
    \end{tabular}
    \caption{\label{tab:compare}
    Comparison of our clustering algorithms with prior works where we omit all $\log$ factors in the additive error, the dependence on privacy parameters, set $\delta_p = 1/n^{1.5}$ and suppress the common $\Delta^2$ factor in the additive error.}
    \end{table}
	
	\subsection{Contributions}
	
	We introduce a differentially private $k$-means clustering algorithm for the global model of differential privacy. This method has an $O(1)$ multiplicative approximation and an $O(\Delta^2 (k \log^2 n\log(1/\delta_p)/\epsilon_p+k\sqrt{d \log(1/\delta_p)}/\epsilon_p ))$ additive approximation guarantee. 
	The additive error is nearly linear in $k$ in contrast with a polynomial in $k$ overhead in previous works, and the multiplicative error is a constant, which is competitive with all previous works. Apart from theoretical value, the algorithm also exhibits an improvement experimentally over earlier work on synthetic and real-world datasets \cite{balcan2017differentially}. For a specific setting of parameters with constants applicable for experiments, we have the following bound. More general asymptotic bounds can be found in the subsequent sections.
	
	\begin{theorem}
	 There is an $(\epsilon_p, \delta_p)$ differentially private algorithm for the $k$-means problem that achieves a utility bound of
	 \begin{align*}
	     O(1) f_D (\OPT_D) + O\left( \frac{k \Delta^2 \log^2 n \log 1/\delta_p}{\epsilon_p} \right) + O\left( \frac{\Delta^2 k \sqrt{d \log 1/\delta_p}}{\epsilon_p} \right), 
	 \end{align*}
	 where $D$ is the input dataset, $f_D (\OPT_D)$ is the optimal $k$-means cost for the input dataset $D$, $d$ is the ambient dimension of $D$, $n$ is the cardinality of $D$, $\Delta$ is the diameter of $D$, and the failure probability of the algorithm is polynomially small in $n$.
	\end{theorem}
	
    For lower bounds, we extend the construction of~\cite{gupta2010differentially} originally for the discrete $k$-medians problem to our setting and show that a linear dependence on $k$ in the additive error is necessary for any finite multiplicative approximation.
    
    \begin{theorem}
        For any $0<\epsilon_p, \delta_p \le 1$ and integer $k$, there is a family of $k$-means instances over the cube $[0,\Delta/\sqrt{d}]^d$ with $d=O(\ln(k/(\epsilon_p\delta_p)))$ dimensions such that the optimal clustering cost is $0$ but any $(\epsilon_p, \delta_p)$-differentially private algorithm would incur an expected cost of $\Omega\left(\frac{\Delta^2 k \ln(\epsilon_p/\delta_p)}{\epsilon_p}\right)$.
    \end{theorem}
    
    The same construction also implies a lower bound for $(\epsilon_p, 0)$-differential privacy.

    \begin{theorem}
        For any $0<\epsilon_p \le 1$ and integers $k$ and $d=\Omega(\ln(k))$, there is a family of $k$-means instances over the cube $[0,\Delta/\sqrt{d}]^d$ such that the optimal clustering cost is $0$ but any $(\epsilon_p, 0)$-differentially private algorithm would incur an expected cost of $\Omega\left(\frac{\Delta^2 k d}{\epsilon_p}\right)$.
    \end{theorem}
    	
	In \cite{gupta2010differentially}, the authors gave an algorithm for solving the discrete version of the problem and subsequent works have focused on identifying a good discretization of the continuous domain and then invoking their algorithm for the discrete case. A recent approach by~\cite{KS18} uses locality sensitive hashing (LSH) to identify a small discrete set of points that serve as potential centers. Inherent in this approach is a trade-off between the multiplicative approximation and the size of this discrete set, which comes from the trade-off in LSH between the approximation and the number of hash functions. The number of discrete candidate centers directly impacts the additive error and thereby brings about a trade-off between the multiplicative and additive errors.
	
	In this work, we avoid this trade-off by going back to the basics and using the most natural approach: discretizing the space using a grid and using all grid points as candidate centers. We can preprocess the data to reduce dimensions to $O((\log n)/\epsilon^2)$ and preserve all distances. However, there can be as many as $(n)^{\log n}$ many points in the grid that we construct since the grid size must start from $1/n$ for negligible additive error. It is not clear how to implement a selection algorithm (such as the exponential mechanism) on such a large number of choices. In fact, it was this hurdle, identified in~\cite{balcan2017differentially}, that prompted subsequent works to find alternative approaches.
	
	An important observation is that a large number of choices is not inherently difficult since it is not hard to sample uniformly among them. Our task is nontrivial since the $k$-means cost objective is a complex function.  
	To simplify the sampling weights, we exploit the connection between clustering and coverage and reduce the problem to finding maximum coverage: count the number of data points within a given radius of each candidate center. The notable advantage is that in maximum coverage, the value of each center is an integer in the range from $1$ to $n$ instead of a real number corresponding to a potential improvement in the $k$-means cost. The centers can hence be partitioned into $n$ classes and we just need to sample uniformly among them. The crucial observation is that there are at most $n^{O(1/\epsilon^4)}$ grid points within the threshold radius of any data point, meaning that there are only a polynomial number of grid points with non-zero coverage. Thus, all but a polynomial number of choices have the same coverage of $0$ making it possible to implement the exponential mechanism in polynomial time.
	
	Given the implementation of the exponential mechanism for coverage, we follow the approach of \cite{JNN} to cover the points using clusters of increasing radii. Note that the approach goes back to the non-private coreset construction of \cite{Chen09}. However, the use of coverage for dealing with each radius has another crucial advantage: like in \cite{JNN}, using the technique of \cite{gupta2010differentially}, the privacy loss only increases by a $\log 1/\delta_p$ factor even though the algorithm has $\Omega(k)$ adaptive rounds of exponential mechanism. 
	
	To obtain a comparable coverage to the optimal solution at each radius $r$, we use a bi-criteria relaxation of maximum coverage and pick many more cluster centers than in the optimal solution. We increase $r$ multiplicatively with factor $(1+\epsilon)$ from $1/n$ to the diameter $2$ and invoke this subroutine for each such $r$, taking the union of all $O(\frac{\log n}{\epsilon})$ sets of candidate centers generated this way to generate a set $C$ of ``good" centers derived from many grids.
	
	By moving each point $p \in D'$ to its closest candidate center $\grid[p]$ in $C$ (with some additional noise), we end up privately constructing a proxy dataset $D''$ with total movement of data points on the order of the $k$-means cost (by virtue of $C$ containing a good $k$-means solution for $D'$). 
	\begin{align*}
		\sum_{p \in D'} d(p, \grid[p]) \simeq \mbox{cost}(D')
	\end{align*}
	
	By the triangle inequality, this bound on the total movement means that any $k$-means solution for $D''$ is immediately a good $k$-means solution for the dataset $D'$ as well, with a constant multiplicative overhead in the cost. A $k$-means solution can be constructed for $D''$ using any non-private $k$-means clustering algorithm, and finally $D$ can be clustered by using noisy averages of clusters in $D'$.
	
	We are able to use a parsimonious privacy budget primarily because of the round independent privacy analysis for the grid-based proxy construction routine using the exponential mechanism which behaves well under composition. There are other privacy preserving noise additions when snapping the data points to the grid for the proxy dataset construction and the final cluster centers but apart from a $\sqrt{d}$ dependence on the ambient dimension they are dominated by the privacy expenditure of the exponential mechanism.
	
	We finish with an experimental evaluation of our algorithm, in which we find that this method performs better than an implementation of previous work~\cite{balcan2017differentially}.

%% file: preliminaries.tex
\section{Preliminaries}
	\subsection{Objective function}
	
	We are given a dataset $D$ of $n$ points that lies in a ball $B_{\Delta/2}(0)$ (the ball of radius $\Delta/2$ centered at $0$) in some high dimensional space $\mathbb{R}^d$. The goal is to find a set of $k$ points $S = \{\mu_1, \dots, \mu_k\}$ such that $\sum_{p \in D} d(p, S)$ is minimal.  Here $d(\cdot, \cdot) : \R^d \times \R^d \to \R$ is the square of the Euclidean distance, that is $d(p,q) := \sum_{i=1}^d (p_i - q_i)^2$. We abuse notation to set $d(p, S) := \min_{\mu \in S} d(p, \mu)$. We define
	\begin{align*}
		f_D (S) = \sum_{p \in D} d(p, S),
	\end{align*}
	so when $S$ is a set of size $k$, $f_D(S)$ is the $k$-means cost of the solution $S$ for the dataset $D$.
	
	\subsection{Differential privacy}
	
	There are a couple of closely related definitions of central differential privacy which can be trivially related to each other. For clarity we disambiguate the situation by uniformly adhering to the formalization in \cite{dwork2014algorithmic}.
	
	\begin{definition}
		We say that two datasets $D, D' \in X^n$ are \emph{neighbouring} if $\lvert D \triangle D'\rvert = 1$, i.e. there is exactly one element in their symmetric difference. We say that an algorithm $A$ is $(\epsilon, \delta)$-differentially private if for any two neighbouring input datasets $D, D'$ and any measurable output set $S$ lying in the co-domain of $A$,
		\begin{align*}
			P(A(D) \in S) \leq e^{\epsilon} P(A(D') \in S) + \delta.
		\end{align*}
	\end{definition}
	
	We now mention some standard tools from the literature of differential privacy.
	
	\begin{lemma}[Exponential Mechanism,~\cite{McSherryT07}]\label{lem:expMech}
		Let input set $D\subset X$, range $R$, and utility function $q: X \times R \rightarrow \mathbb{R}$. 
		The {\em Exponential Mechanism} $M_E (D,q,\epsilon_E)$ with privacy parameter $\epsilon_E$ outputs an element $r\in R$ sampled according to the distribution
		\begin{align*} 
			P(M_E (D,q,\epsilon_E) = r) := \exp\left(\frac{\epsilon_E \cdot q(D,r)}{2\Delta q}\right),
		\end{align*}
		where $\Delta q$ is the sensitivity of the utility $q$; i.e. 
		\begin{align*}
		    \Delta q = \max_{\substack{ r \in R \\ \lvert A \Delta B \rvert = 1}} \lvert q(A,r) - q(B,r) \rvert.
		\end{align*}
		The Exponential Mechanism is $(\epsilon_E,0)$-differentially private and with probability at least $1-\gamma$,
		\begin{align*}
    		\lvert \max\limits_{r\in R} q(D,r) - q(D, M_E (D,q,\epsilon_E))\rvert \leq \frac{2\Delta q}{\epsilon_E}\log\left(\frac{\lvert R\rvert }{\gamma}\right).
		\end{align*}
	\end{lemma}
	
	\begin{lemma}[Laplace mechanism, \cite{dwork2006calibrating}]\label{lem:lapMech}
		Given any function $f:\mathbb{N}^{\lvert X \rvert} \to \mathbb{R}^k$, the Laplace mechanism with privacy parameter $\epsilon_L$ is defined as
		\begin{align*}
			M_L (x, f(\cdot), \epsilon_L) := f(x) + (Y_1, \dots, Y_k)
		\end{align*}
		where $Y_i$ are i.i.d. $\sim \Lap \left( \frac{\Delta f}{\epsilon_L}\right)$, $\Lap (x|b) = \frac{1}{2b} \exp\left(-\frac{\lvert x \rvert}{b} \right)$ and $\Delta f = \max_{\lvert X \triangle Y\rvert = 1} \lvert f(X) - f(Y) \rvert$, i.e. the $\ell_1$ sensitivity of $f$ over all pairs of neighbouring datasets. The Laplace mechanism is $(\epsilon_L, 0)$-differentially private.
	\end{lemma}
	
	\begin{theorem}[Basic composition, \cite{dwork2009differential}]\label{thm:basicComp}
		If a sequence of algorithms $M_i$ with $(\epsilon_i, \delta_i)$-differential privacy guarantees are composed in order then the composite process satisfies $(\sum_{i=1}^m \epsilon_i , \sum_{i=1}^m \delta_i)$-differential privacy.
	\end{theorem}
	
	We state the parallel composition theorem from \cite{mcsherry2009privacy}, slightly modifying the proof to extend this result to be able to use $(\epsilon, \delta)$-differentially private subroutines where $\delta\not= 0$.
	
	\begin{theorem}[Parallel composition, \cite{mcsherry2009privacy}]\label{thm:parallelComp}
		Let $M_i$ each provide $(\epsilon_i, \delta_i)$-differential privacy. Let $\{ D_i : i \in \Lambda \}$ be arbitrary disjoint subsets of the input domain $D$. For input dataset $X$, the sequence of $M_i (X \cap D_i)$ provides $(\max_i \epsilon_i, \max_i \delta_i)$differential privacy.
	\end{theorem}
	
	\begin{proof}
		Let $A$ and $B$ be neighbouring datasets, $A_i = A \cap D_i$ and $B_i = B \cap D_i$ for $i \in \Lambda$. Let $M_i$ be $(\epsilon_i, \delta_i)$ differentially private subroutines for $i \in \Lambda$. Since $\lvert A \triangle B \rvert = 1$, there is at most one partition set $D_{i*}$ of the domain such that $A_{i*} \not= B_{i*}$. We bound the ratio of probabilities of the output tuple as
		\begin{align*}
			\frac{P((M_i(A_i))_{i \in \Lambda} = (S_i)_{i \in \Lambda)}}{P((M_i(B_i))_{i \in \Lambda} = (S_i)_{i \in \Lambda})} &= \prod_{i\in\Lambda} \frac{P(M_i(A_i) = S_i}{P(M_i(B_i) = S_i} \\
			&= \frac{P(M_{i^*} (A_{i^*}) = S_{i^*})}{P(M_{i^*}(B_{i^*}) = S_{i^*})}
		\end{align*}
		Since $M_{i^*}$ is $(\epsilon_{i^*}, \delta_{i^*})$-differentially private, it follows that with probability $1 - \delta_{i^*}$, this ratio is bounded by $\exp(\epsilon_i^*)$. Since this is true for every pair of neighbouring datasets, we can summarise this by saying that this sequence of operations is $(\max_i \epsilon_i, \max_i \delta_i)$- differentially private.
	\end{proof}
	
	\subsection{Private averaging}
	
	Here we describe the noisy averaging algorithm from \cite{NSV16} which we use in the last step to derive cluster centers privately from cluster IDs and data without revealing the cluster member vectors or the cluster size. The authors of that work use the Gaussian mechanism in conjunction with a bound on the sensitivity of the averaging operator and addition of Laplace noise to the number of points in the cluster to ensure privacy. This additional step of adding Laplace noise is necessary because in this setting the cluster sizes are derived from sensitive information, and since the Gaussian noise added to mask the exact mean depends upon the cluster size, the noise parameter must itself be masked by additional noise (which is data-independent). 
	
	In \cite{NSV16}, the authors use a slightly different definition of privacy where datasets $D, D'$ are considered neighboring if $D = D' \backslash \{p \} \cup \{ p'\}$ for some $p \in D'$. The statements here are with the privacy parameters modified to fit the definition of differential privacy that we are working with; $\Delta$ is any upper bound on the diameter of the input set.
	
	\begin{algorithm}
		\caption{NoisyAVG\cite[Algorithm 5]{NSV16}}
		\label{alg:noisyAVG}
		\KwData{Multiset $V$ of vectors in $\mathbb{R}^d$, predicate $g$, parameters $\epsilon,\delta$}
		Set $\hat{m} = \lvert \{ v \in V: g(v) = 1 \}\rvert + \Lap (5/\epsilon) - \frac{5}{\epsilon} \ln(2/\delta)$. If $\hat{m} < 0$, output a uniformly random point in the domain $B_{\Delta/2}(0)$.\\
		Denote $\sigma = \frac{5 \Delta}{4\epsilon \hat{m}} \sqrt{2 \ln(3.5/\delta)}$, and let $\eta \in \mathbb{R}^d$ be a random noise vector with each coordinate sampled independently from $N(0,\sigma^2)$. \\
		\KwRet{$g(V) + \eta$}
	\end{algorithm}
	
	\begin{theorem} [Privacy {\cite[Theorem A.3]{NSV16}} and noise {\cite[Observation A.1]{NSV16}} bounds for \cref{alg:noisyAVG}] \label{thm:noisyAVG}
	    \Cref{alg:noisyAVG} is an $(\epsilon, \delta)$-differentially private algorithm for $\epsilon\le 1/3$. Further, if $V$ and $g$ are such that $m = \lvert \{v\in V : g(v) = 1 \} \rvert \ge A\left(\frac{1}{\epsilon} \ln \left(\frac{1}{\beta \delta}\right)\right)$ for sufficiently large constant $A$, then with probability $1-\beta$, \cref{alg:noisyAVG} returns $g(V) + \eta$ where $\eta$ is a vector where every coordinate is sampled i.i.d. from $N(0,\sigma^2)$ for some $\sigma \leq \frac{4 \Delta}{\epsilon m} \sqrt{2 \ln(8/\delta)}$, where $\Delta$ is the diameter of the input set.
	\end{theorem}
	
	\subsection{Other technical tools}
	
	In this subsection we list some technical lemmata that we will find useful to refer to in the main body of this work.
	
	\begin{lemma}[Greedy set cover bicriteria solution guarantee]\label{lem:maxCoverRelax}
		Let there be a set of elements $U$, a family of sets $\mathcal{S} \subset 2^U$, and the promise that there is some subfamily of sets $\Z \subset \mathcal{S}$ that covers $\U$. Suppose that for $Y = 2 \lceil \lvert \Z \rvert \log 1/\epsilon \rceil + 1$ we iteratively pick a collection of sets $\C = \{c_1, \dots, c_Y\} \subset \mathcal{S}$. We denote the set of elements not picked by the $i$th iteration $U_1 = U$ and $U_i = U_{i-1} \backslash c_{i-1}$. If $c_i \cap U_i$ is at least half as big as $\max_{c \in \mathcal{S}} (c \cap U_i)$ for all $i$, then $\mathcal{S}$ will cover $(1-\epsilon)$ of all elements in $U$. Formally, if
		\begin{align*}
		\C &= \{c_i : c_i \in \mathcal{S}, i = 1, \dots, \lceil 2 \lvert\Z \rvert \log 1/\epsilon \rceil \},\\
		U_i &= U \backslash ( c_1 \cup \dots \cup c_{i-1} )
		\end{align*}
		where 
		\begin{align*}
			\left|c_i \cap U_i \right| \geq \frac{\max\limits_{c \in S} \left( c \cap U_i \right) }{2} ,
		\end{align*}
		then $|\bigcup_{i} c_i | \geq (1-\epsilon) |U|$.
	\end{lemma}
	
	\begin{proof}
		The idea behind this proof is simple; it suffices to show that in every iteration we always cover a certain fraction of the thus far uncovered elements. Telescoping this multiplicative guarantee will give us our desired result. Formally, we know that since $\bigcup_{z \in \Z} z = U $, the union $\bigcup_{z \in \Z} z$ also covers the set of unpicked elements $U_{i-1}$. Enumerating the elements of $U_i$ by summing the cardinalities of its intersections with the members of $\Z$, we get
		\begin{align*}
			\sum_{z \in \Z} \lvert z\cap U_{i-1} \rvert & = \lvert U_{i-1} \rvert \\
			\Rightarrow \max_{z \in \Z} \left\lvert z\cap U_{i-1} \right\rvert & \geq \frac{\left\lvert U_{i-1} \right\rvert}{\left\lvert \Z \right\rvert } \\
			\Rightarrow \left\lvert c_i \cap U_{i-1} \right\rvert &\geq \frac{\left\lvert U_{i-1} \right\rvert}{2 \left\lvert \Z \right\rvert } \\
			\Rightarrow \left\lvert c_i \right\rvert &\geq \frac{\left\lvert U_{i-1} \right\rvert}{2 \left\lvert \Z \right\rvert }
		\end{align*}
		Since $U_{i-1} = U_i \sqcup c_i$, 
		\begin{align*}
			\lvert U_i \rvert &\leq \left(1 - \frac{1}{2 \lvert \Z \rvert } \right) \lvert U_{i-1} \rvert \\
			\Rightarrow \lvert U_i \rvert &\leq \left(1 - \frac{1}{2 \lvert \Z \rvert } \right)^{i-1} \lvert U \rvert
		\end{align*}
		It follows that for 
		\begin{align*}
			i-1 > \log_{1 - \frac{1}{2 \lvert \Z \rvert}} \epsilon = \frac{\log \epsilon}{\log 1 - \frac{1}{2 \lvert \Z \rvert}} > \frac{\log \epsilon}{- \frac{1}{2 \lvert \Z \rvert}} = 2 \lvert \Z \rvert \log 1/\epsilon,
		\end{align*}
		$\lvert U_Y \rvert < \epsilon \lvert U \rvert$ and so the size of the complement $\lvert c_1 \cup \dots \cup c_Y \rvert$ is $\geq (1-\epsilon) \lvert U \rvert$.
	\end{proof}
	
	\begin{lemma}\label{lem:triangle}
		In any metric space with a metric $d(\cdot,\cdot)$ and $p\geq 1$, $d^p(a,b) \leq 2^{p-1} (d^p(a,c) + d^p (c,b))$.
	\end{lemma}
	
	\begin{proof}
		Applying Jensen's inequality with the function $g(x) = x^p$ we have
		\begin{align*}
			\left( \frac{d (a,c)}{2} + \frac{d (c,b)}{2} \right)^p \leq \frac{d^p (a,c) + d^p (c,b)}{2}.
		\end{align*}
		Since $d$ is a metric, by the triangle inequality
		\begin{align*}
			d(a,b) &\leq d(a,c) + d(c,b)\\
			\Rightarrow \frac{d^p(a,b)}{2^p} & \leq \left( \frac{d (a,c)}{2} + \frac{d (c,b)}{2} \right)^p \\
			&\leq \frac{d^p (a,c) + d^p (c,b)}{2} \\
			\Rightarrow d^p (a,b) &\leq 2^{p-1} (d^p(a,c) + d^p (c,b))
		\end{align*}
	\end{proof}
	
	\begin{lemma}[Concentration bound for privacy analysis, \cite{gupta2010differentially}]\label{lem:gupta}
		Let $R_1, \dots, R_n$ be some Bernoulli random variables, $R_i \sim \mbox{Ber}(p_i)$, i.e. $R_i = 1$ with probability $p_i$ and $0$ with probability $1-p_i$, where $R_i$ may depend arbitrarily on $R_1, \dots R_{i-1}$. Let $Z_j = \prod_{i=1}^j (1 - R_i)$. Then
		\begin{align*}
			P\left(\sum_{i=1}^n p_i Z_i > q \right) \leq \exp(-q).
		\end{align*}
	\end{lemma}

%% file: lower_bounds.tex
\section{Lower bounds}\label{sec:lowerBounds}

Following the construction in theorem 4.4 of \cite{gupta2010differentially}, we derive lower bounds for the $k$-means clustering problem in the $(\epsilon,\delta)$ and $(\epsilon,0)$-differential privacy regimes. 

\begin{theorem}
        For any $0<\epsilon_p, \delta_p \le 1$ and integer $k$, there is a family of $k$-means instances over the cube $[0,\Delta/\sqrt{d}]^d$ with $d=O(\ln(k/(\epsilon_p\delta_p)))$ dimensions such that the optimal clustering cost is $0$ but any $(\epsilon_p, \delta_p)$-differentially private algorithm would incur an expected cost of $\Omega\left(\frac{\Delta^2 k \ln(\epsilon_p/\delta_p)}{\epsilon_p}\right)$.
    \end{theorem}
    
	\begin{proof}
	    Let the ambient dimension $d=\Theta(\ln(k/((e^{\epsilon_p}-1)\delta_p))$ and $W$ be the set of codewords of an error correcting code with constant rate and constant relative distance in $\{0,1\}^d$. The dimension $d$ and codewords $W$ are chosen so that $|W| \ge k/((e^{\epsilon_p}-1)\delta_p)$. Let $L = \ln((e^{\epsilon_p}-1)/(4\delta_p))/(2\epsilon_p)$. Our input domain is the unit cube $[0,1]^d$ with diameter $\Delta=\sqrt{d}$. Note that for other values of $\Delta$, we can simply re-scale the construction.
	    
	    Suppose $M$ is an arbitrary $(\epsilon_p, \delta_p)$-differentially private algorithm that on input $D \subset [0,1]^d$ outputs a set of $k$ locations. Let $M'$ be the algorithm that first runs $M$ on the input and then snaps each output point to the nearest point in $W$. By post-processing, $M'$ is $(\epsilon_p, \delta_p)$-differentially private. Furthermore, observe that if the input points are located at a subset of $W$ then the cost of $M'$ is within a factor $4$ of the cost of $M$. Let $A$ be a size $k$ subset of $W$ chosen uniformly at random and the dataset $D_A$ is a multiset containing each point in $A$ with multiplicity $L$. Note that the optimal cost for $D_A$ is $0$.
	    
	    We would like to analyze $\phi = \E_{A,M'}[|A \cap M'(D_A)|]/k$. We have:
	    \begin{align*}
	    	k\phi &= \E_{A, M'} \left[ \sum_{i\in A} 1_{i\in M'(D_A)}\right]  \\
	    	&= k\E_{A, M'}\E_{i\in A}[1_{i\in M'(D_A)}]\\
	    	&= k\E_{i\in W}\E_{A, M'}[1_{i\in M'(D_A)} | i\in A]
	    \end{align*}
	    Let $i'$ be an random point in $W$ not in $A$. Changing $A$ to $A' = A\setminus\{i\} \cup \{i'\}$ requires changing $2L$ elements of $D_A$. Notice that for random $A\setminus\{i\}$ in $W\setminus\{i\}$ and random $i'$ in $W\setminus A$, we have that $A'$ is still a uniformly random subset of $W\setminus\{i\}$.
	    Thus, 
	    \begin{align*}
            \E_{i\in W}\E_{A', M'}[1_{i\in M'(D_{A'})} | i\not\in A'] \ge \left(\E_{i\in W}\E_{A, M'}[1_{i\in M'(D_{A})} | i\in A]\right)\exp(-\epsilon_p \cdot 2L) -  \frac{\delta_p}{e^{\epsilon_p}-1}
	    \end{align*}
	    Here we use the fact that $M'$ is $(\epsilon_p, \delta_p)$-differentially private, and that the $\delta_p$ losses in expectation decrease geometrically with factor $\exp(-\epsilon_p)$ so the net leakage from the $\delta$ term can be lower bounded by the sum of an infinite geometric progression. Continuing,
	    \begin{align*}
	    	\E_{i\in W}\E_{A', M'}[1_{i\in M'(D_{A'})}] &\ge \phi \exp(-\epsilon_p \cdot 2L)- \delta_p/(e^{\epsilon_p}-1) \\
	    	&\ge 4\phi\delta_p/(e^{\epsilon_p}-1) -\delta_p/(e^{\epsilon_p}-1)
	    \end{align*}
	    
	    Since the output $M'(D_{A'})$ has at most $k$ points, the LHS is at most $k/|W|$. Thus, $\phi \le (k/|W| +\delta_p/(e^{\epsilon_p}-1))/(4\delta_p/(e^{\epsilon_p}-1)) \le 1/2$. 
	    
	    For each point in $A\setminus M'(D_A)$, the algorithm incurs a cost of $\Theta(L\Delta^2)$ due to the multiplicity of $L$ of points in $D_A$ and the fact that all points in $W$ are at distance $\Theta(\Delta)$ apart. Therefore, the expected cost of $M'$, and consequently the cost of $M$, is $\Omega(kL\Delta^2)=\Omega\left(\frac{\Delta^2 k\ln(\epsilon_p/\delta_p)}{\epsilon_p}\right)$.
	\end{proof}
	
	By the same proof, one can also obtain a lower bound for $(\epsilon_p,0)$-differential privacy.
    \begin{theorem}
        For any $0<\epsilon_p \le 1$ and integers $k$ and $d=\Omega(\ln(k))$, there is a family of $k$-means instances over the cube $[0,\Delta/\sqrt{d}]^d$ such that the optimal clustering cost is $0$ but any $(\epsilon_p, 0)$-differentially private algorithm would incur an expected cost of $\Omega\left(\frac{\Delta^2 k d}{\epsilon_p}\right)$.
    \end{theorem}
    	
	\begin{proof}
	    Let $W$ be the set of codewords of an error correcting code with constant relative rate and constant relative distance in $\{0,1\}^d$. Note that $|W| = 2^{\Omega(d)}$. Let $L = \ln(|W|/(2k))/(2\epsilon_p)$. Our input domain is the unit cube $[0,1]^d$ with diameter $\Delta=\sqrt{d}$. Note that for other values of $\Delta$, we can simply re-scale the construction.
	    
	    Suppose $M$ is an arbitrary $(\epsilon_p, 0)$-differentially private algorithm that on input $D \subset [0,1]^d$ outputs a set of $k$ locations. Let $M'$ be the algorithm that first runs $M$ on the input and then snaps each output point to the nearest point in $W$. By post-processing, $M'$ is $(\epsilon_p, 0)$-differentially private. Furthermore, observe that if the input points are located at a subset of $W$ then the cost of $M'$ is within a factor $4$ of the cost of $M$. Let $A$ be a size $k$ subset of $W$ chosen uniformly at random and the dataset $D_A$ is a multiset containing each point in $A$ with multiplicity $L$. Note that the optimal cost for $D_A$ is $0$.
	    
	    We would like to analyze $\phi = \E_{A,M'}[|A \cap M'(D_A)|]/k$. We have:
	    \begin{align*}
	    	k\phi &= \E_{A, M'} \left[ \sum_{i\in A} 1_{i\in M'(D_A)}\right]  \\
	    	&= k\E_{i\in W}\E_{A, M'}[1_{i\in M'(D_A)} | i\in A]
	    \end{align*}
	    Let $i'$ be an arbitrary point in $W$ not in $A$. Changing $A$ to $A' = A\setminus\{i\} \cup \{i'\}$ requires changing $2L$ elements of $D_A$. Thus,
	    \[
	    	\E_{i\in W}\E_{A', M'}[1_{i\in M'(D_{A'})}] \ge \phi \exp(-\epsilon_p \cdot 2L)
	    \]
	    
	    Since the output $M'(D_{A'})$ has at most $k$ points, the LHS is at most $k/|W|$. Thus, $\phi \le (k/|W|) \exp(2L\epsilon_p) \le 1/2$. 
	    
	    For each point in $A\setminus M'(D_A)$, the algorithm incurs a cost of $\Theta(L\Delta^2)$ due to the multiplicity of $L$ of points in $D_A$ and the fact that all points in $W$ are at distance $\Theta(\Delta)$ apart. Therefore, the expected cost of $M'$, and consequently the cost of $M$, is $\Omega(kL\Delta^2)=\Omega\left(\frac{\Delta^2 kd}{\epsilon_p}\right)$.
	\end{proof}

%% file: algorithm.tex
\section{Algorithm}
	
	We introduce some notation to make the analysis of \cref{alg:privateKMeans,alg:privateGridSetCover} easier.
	\begin{itemize}
		\item $\OPT_D$, optimal solution: We let $\OPT_D$ be the lexicographically minimal optimal $k$-means set for the dataset $D$. The lexicographic minimality is just for uniqueness, it has no other significance.
		\item $\epsilon$, multiplicative approximation constant: We let $\epsilon$ be an approximation constant that is used in the Johnson-Lindenstrauss transform and such that $1 + \epsilon$ is the factor the grid unit length and threshold increase by in each iteration of the loop on \cref{alg:privateKMeans;line:loopStart} to \cref{alg:privateKMeans;line:loopEnd}. We will require that $\epsilon \in (0,1)$ and be bounded away from $1$, say $\epsilon\leq 0.5$.
		\item $m$, number of iterations: We let $m = \lceil \log_{1 + \epsilon} 2n \rceil$, the total number of iterations for which \cref{alg:privateGridSetCover} is called.
		\item $r_i$, threshold radii: For $i\in \{1,\dots ,m\}$, we let $r_i = \frac{t_i \sqrt{d}}{\epsilon} = \frac{(1 + \epsilon)^{i-1}}{n}$, the $i$th threshold radius used for computing the max cover bi-criteria relaxation. For notational convenience we set $r_0 = 0$, and note that $r$ increases geometrically from $r_1 = \frac{1}{n}$ to $r_m = 2$.
		\item $G_i, t_i$, grid and unit length: For $i\in \{1,\dots ,m\}$, we let $G_i = \{-1, -1 + t_i, -1 + 2t_i , \dots , 1 - t_i, 1\}^{d'}$, where $t_i = \frac{\epsilon}{n\sqrt{d}} (1+\epsilon)^{i-1}$, the grid unit length in the $i$th iteration. Note that $\lvert G_i \rvert = \lfloor \frac{1}{t_i} \rfloor^{d'}$.
		\item $\snap{\cdot}{i}$, floor to grid function: We let $\snap{v}{i}$ for any vector $v \in \mathbb{R}^d$ denote $((t_i \lfloor \frac{v_1}{t_i} \rfloor), \dots, (t_i \lfloor \frac{v_{d'}}{t_i} \rfloor))$, i.e. $\snap{v}{i}$ is the coordinate-wise ``floor" of $v$ in the grid of unit length $t_i$. 
		\item $o_i$, ideal thresholded objectives: For $i\in \{1,\dots ,m\}$, we let $o_i = \{p \in D' : d(p,\OPT_{D'}) \in [r_{i-1},r_i) \}$. Since $D' \subset B_{1} (0)$, $D' = \sqcup_{i=1}^m o_i$
		\item $a_i$, set of points covered in $i$th call: For $i\in \{1,\dots ,m\}$, we let $a_i = B_{r_i + t_i \sqrt{d'}}(C_i) \cap D'$ where $C_i$ is the set of points returned by \cref{alg:privateGridSetCover} when called in the $i$th iteration of \cref{alg:privateKMeans}.
		\item $B_r (\cdot)$: We let $B_r (S)$ be the union of all balls of radius $r$ whose center is an element of $S$. We abuse notation to let $B_r (g) := B_r (\{g\})$ and observe that $|B_r (\cdot)|$ is a monotonic positive submodular function for any $r\geq 0$, as is $\lvert T \cap B_r(\cdot) \rvert$, for any fixed set $T$.
		\item $\cover$: We let $\cover[g]$ denote that set of uncovered data points that are covered within the radius $(r_i + t_i \sqrt{d'})$ around $g$ in the $i$th call to \cref{alg:privateGridSetCover}. This is a subset of the data points in $B_{(1+\epsilon)r_i}$
		\item $\grid$: We let $\grid[p]$ denote the grid point that covers $p$ in the call to \cref{alg:privateGridSetCover} that removes $p$ from $D'_1$.
		\item $\kCenter$: We let $\kCenter[p]$ denote the closest element of $\OPT_{D'}$ to the datapoint $p \in D'$.
	\end{itemize}
	
	At a high level the algorithm can be described in four steps. 
	\paragraph{Step 1:}{First the dataset $D \subset B(0,\Lambda/2) \subset \mathbb{R}^d$ is preprocessed via dimension reduction, scaling and projection to produce a dataset $D' \subset B_1(0) \subset \mathbb{R}^{d'}$ where $d=O((\log n)/ \epsilon^2 )$. Note that with high probability we do not need to project any point and so need not account for it in the privacy analysis; however, by projecting instead of re-scaling, we preserve privacy. To start with a finite number of candidate centers we construct multi-dimensional grids of side lengths $t_i$ and observe that if $\mu$ is a center of a cluster with radius $r_i$ in the optimal solution, then by the triangle inequality a ball of radius $r_i + t_i \sqrt{d'}$ centered at $\snap{\mu}{i}$ (the ``floor" of $\mu$ the in grid) contains all the points of the same cluster.}
	\paragraph{Step 2:}{Next, for geometrically increasing grid unit lengths $t_i$ with growth factor $(1+\epsilon)$ starting from $\epsilon/n\sqrt{d'}$ and increasing to $2/\sqrt{d'}$ we create grids and identify possible centers of clusters with radius in the interval $[r_{i-1}, r_i)$. This is done by counting the number of datapoints within $r_i + t_i \sqrt{d'}$ of every grid point. To ensure a polynomial time method this enumeration is done by iterating over datapoints and adding each datapoint to the grid points which could be valid cluster centers - we do this by keeping in hand a set of valid offsets $V_i$ and simply incrementing counts for all grid points within an offset of $\snap{p}{i}$. Since we are looking for a $k$-means solution there could be as many as $k$ clusters for any given radius, which requires us to greedily identify the $k/\epsilon$ best grid points to obtain close to optimal coverage (see \cref{lem:maxCoverRelax}). We take the union of all $\log_{1 + \epsilon} 2/(1/n) = O((\log n)/\epsilon)$ sets of $k/\epsilon$ points so found to construct the set $C$.}
	\paragraph{Step 3:}{Once this set $C$ containing a good cluster solution is identified, the idea is to construct the proxy dataset $D''$ by moving each point to its closest point in $C$. However, constructing the proxy dataset in this way means accessing the sensitive data again. In order to maintain privacy in this step, instead of directly moving datapoints to points in $C$, we compute the counts $n_c$ of the number of datapoints that would ideally be moved to $c$ and add Laplace noise to $n_c$ to get $\tilde{n}_c$. $D''$ then contains $\tilde{n}_c$ copies of $c$ for all $c \in C$.}
	\paragraph{Step 4:}{In the final step we apply any non-private $k$-means clustering algorithm to $D''$ to get some cluster centers $S''$. We cluster $D'$ using these cluster centers to get clusters $C'$, and define final clusters for $D$ by identifying points with their images under the Johnson Lindenstrauss map. Since this step again uses sensitive data we use the Gaussian mechanism to return noisy averages of these clusters $C'$ for the dataset $D$ to derive the set of $k$-means $S$.}
	\vspace{\baselineskip}
	
	\begin{algorithm}
		\caption{Private $k$-means}
		\label{alg:privateKMeans}
		\KwData{$D \subset \mathbb{R}^d$ dataset, $\lvert D' \rvert = n$.}
		\KwResult{$S = \{ \tilde{\mu}_1, \dots, \tilde{\mu}_k\} \subset \mathbb{R}^d$}
		$T \sim \mbox{JohnsonLindenstrauss}(n, \epsilon)$\\
		$D' \leftarrow T(D) $  \\
		$d' \leftarrow \mbox{dim}(T) = O((\log n)/\epsilon^2)$ \\
		Scale $D'$ down by a factor of $\frac{\Delta}{2 (1+\epsilon)}$ and project to $B_1 (0)$ \\
		Let $T'$ be the composition of $T$ with the scaling and projection so that $T'(D) = D'$\\
		$r_1 \leftarrow 1/n$\\
		$t_1 \leftarrow \epsilon/(n \sqrt{d'})$ \\
		\For{$i = 1, \dots, m = \lceil \log_{1 + \epsilon} 2n \rceil $}{\label{alg:privateKMeans;line:loopStart}
			$C_i \leftarrow \cref{alg:privateGridSetCover}(D', t_i, r_i)$ \\
			$r_{i+1} \leftarrow (1 + \epsilon)r_i$.\\
			$t_{i+1} \leftarrow (1 + \epsilon)t_i$.\\
		}\label{alg:privateKMeans;line:loopEnd}
		$D' \leftarrow T'(D)$ \tcp*[r]{resetting the dataset to account for points lost in call to \cref{alg:privateGridSetCover}}
		$C = \bigcup_{i=1}^m C_i$\\
		Assign all points in $D'$ to their closest point $c \in C$\\
		Let $n_c$ be the number of points in $D'$ assigned to $c$\\
		For each $c \in C$ set $n'_c = n_c + \mbox{Lap}\left(\frac{1}{\epsilon_L}\right)$\label{alg:privateKMeans;line:laplaceMech}\\
		Let $D''$ be the dataset where every $c \in C$ is repeated $n'_c$ times\\
		$S'' = \{\mu''_1, \dots, \mu''_k\} \leftarrow \mbox{Lloyd}(D'')$\\
		$ D'_i \leftarrow \{ p \in D' : \argmin_{\mu'' \in S''} d(p, \mu'') = \mu''_i \}$ for $i = 1, \dots, k$\label{alg:privateKMeans;line:cluster}\\
		\For{$i = 1,\dots, k$}{
		$\tilde{\mu}_i = \cref{alg:noisyAVG} (D, 1_{D'_i}, \epsilon_G, \delta_G)$\label{alg:privateKMeans;line:gaussMech} \tcp*[r]{$1_{D'_i} (p)$ indicates whether $T'(p)\in D'_i$ for $p\in D$}
		} 
		\KwRet{$\tilde{S} = \{ \tilde{\mu}_1, \dots, \tilde{\mu}_k\}$}
	\end{algorithm}
	
	\begin{algorithm}
		\caption{Private grid set cover}
		\label{alg:privateGridSetCover}
		\SetKwFor{RepTimes}{repeat}{times}{end}
		\KwData{$D'$ dataset (passed by reference), $t_i$ grid unit length, $r_i$ threshold radius}
		\KwResult{set $C_i \subset G_i$}
		$C_i \leftarrow \emptyset$ \\
		\RepTimes{$k'$}{
			$\cover \leftarrow$ empty linked list \\
			$V_i \leftarrow \{ v:v \in \mathbb{N}^{d'},\: \sum_{j=1}^{d'} (t_i v_j)^2 < (r_i + t_i \sqrt{d'})^2 \}$\\
			\For{all $p \in D'$}{
				\For{all $v \in V_i$}{
					\For{all $s \in \{0,1\}^{d'}$}{
						$t_i b = \snap{p}{i} + t_i s + (2s - \bar{1}) t_i v$ \tcp*[r]{where $\bar{1}$ is the all-ones vector}
						\If{$d(t_i b, p) < (r_i + t_i \sqrt{d'})^2 $}{
						$\cover[t_i b] \pluseq \{p\}$\\
					}
					}
				}
			}
			$\totalCover \leftarrow 0$\\
			\For{$g \in \cover$}{
				$\totalCover \pluseq \exp \left(\frac{\epsilon_E \lvert \cover[g] \rvert}{2}\right)$\\
			}
			$\totalCover \pluseq \lvert G_i \rvert - len[\cover]$ \\
			Let $P_{samp} = 1 - \frac{\lvert G_i \rvert}{\totalCover}$.\\
			\uIf{$\mbox{Ber}\left(P_{samp} \right) = 1$}{
				$g \leftarrow$ pick $i \in [len[\cover]] \sim P(g) \propto \: \exp\left(\frac{\epsilon_E \lvert \cover[g] \rvert}{2} \right) -1$
			}
			\Else{
				$g \leftarrow$ pick $i$ uniformly at random from $G_i$
			}
			$C_i \leftarrow C_i \cup \{g\}$\\
			$D' \leftarrow D' \backslash \cover[g]$\\
		}
		\KwRet{$C_i$}
	\end{algorithm}
	
	The formal pseudocode \cref{alg:privateGridSetCover} requires some additional justification; the construction of the offset set $V_i$, and the polynomial time implementation of the exponential mechanism.
	
	\begin{claim} \label{clm:nbdBound}
		A data point $p$ is within distance $r_i + t_i \sqrt{d'}$ of a grid point $t_i b$ for $b \in \mathbb{Z}^{d'}$ only if
		\begin{align*}
			\sum_{j=1}^{d'} \min ( (\snap{p}{i}_j - t_i b_j)^2, (\snap{p}{i}_j - t_i (b_j + 1) )^2) \leq (r_i + t_i \sqrt{d'})^2.
		\end{align*}
		Let $V_i = \{ v  : v \in \mathbb{N}^{d'}, \sum_{j=1}^{d'} t_i^2 v_j^2 < (r_i + t_i \sqrt{d'})^2\}$. If $t_i b$ is a grid point such that $d(p, t_i b) < (r_i + t_i \sqrt{d'})^2$ then for some $s \in \{ 0, 1\}^d$ and $v \in V_i$, $t_i b = \snap{p}{i} + t_i s + (2s - \bar{1}) t_i v$, where $\bar{1} = (1,1,\dots, 1)$, the $d'$-dimensional all-ones vector.
	\end{claim}
	
	\begin{proof}
		Informally, $p_j$ is a real number and the $t_i b_j$ lie on regularly spaced intervals on the number line. Since $\snap{p}{i}_j$ and $\snap{p}{i}_j + t_i$ are the two closest grid points to $p$, any grid point must be closer to one of these neighbours than it is to $p_j$.
		
		If $p$ is within distance $r_i + t_i \sqrt{d'}$ of a grid point $t_i b$ for $b \in \mathbb{Z}^{d'}$ then by definition
		\begin{align*}
			\sum_{j=1}^{d'} (p_j - t_i b_j)^2 \leq (r_i + t_i \sqrt{d'})^2.
		\end{align*}
		If $p_j \geq t_i b_j$ then $p_j = t_i b_j + t_i x + y$ for some $x \in \mathbb{N}$ and $y \in [0,1)$. Since $\snap{p}{i}_j = t_i b_j + t_i x$, $(\snap{p}{i}_j - t_i b_j)^2 < (p_j - t_i b_j)^2$. Else, if $p_j < t_i b_j$ then $p_j = t_i b_j - t_i x - y$, with same ranges for $x$ and $y$. Then $\snap{p}{i}_j + t_i = t_i b_j - t_i x$ so $(\snap{p}{i}_j - t_i b_j + t_i)^2 < (p_j - t_i b_j)^2$. Therefore we have that $\min((\snap{p}{i}_j - t_i b_j)^2,(\snap{p}{i}_j - t_i (b_j + 1))^2) < (p_j - t_i b_j)^2$. Summing up this inequality over the index $j$ and using the display above gives us the desired result.

		Let $s_j = 0$ if $p_j \geq t_i b_j$ and $s_j = 1$ if $p_j < t_i b_j$. Tracing the proof of the first half and letting $v$ be such that $t_i v_j = \min ((\snap{p}{i}_j - t_i b_j)^2,(\snap{p}{i}_j - t_i (b_j + 1) )^2)$, it follows that $t_i b_j = \snap{p}{i}_j + t_i s_j + (2 s_j - 1) v_j$ for all $j \in [k]$. Putting together all coefficients this gives us that $t_i b = \snap{p}{i} + t_i s + (2s - \bar{1}) t_i v$ for some $v \in V_i$.
	\end{proof}
	
	\begin{claim}
		After computing the cover of each grid point, \cref{alg:privateGridSetCover} executes the exponential mechanism correctly and in polynomial time.
	\end{claim}
	
	\begin{proof}
		First we note that there are at most polynomially many grid points whose cover is updated in any call to \cref{alg:privateGridSetCover}. From \cref{clm:nbdBound} we know that for any data point the only grid points whose cover must be updated lie in $V_i$. It will hence suffice to show that $\lvert V_i \rvert < n^{O(1/\epsilon^4)}$. To get the number of unsigned $d'$-dimensional ordered tuples $v$ for which $\sum_i t_i^2 v_i^2 < (r_i + t_i \sqrt{d'})^2 \Leftrightarrow \sum_i v_i^2 < d' (\frac{1}{\epsilon} + 1)^2$, it suffices to count the number of ways of partitioning $d' (\frac{1}{\epsilon} + 1)^2 + d' + 1$ balls into $d'+1$ distinguishable bins. We can do this by placing the balls in a line and choosing $d'$ gaps between them. It follows that $\lvert V \rvert= 2^{d'} \binom{d' (\frac{1}{\epsilon} + 1)^2 + d' + 1 }{d' + 1} = O\left(2^{\frac{d'}{\epsilon^2}}\right) < n^{O(1/\epsilon^4)}$, using that $d' = O\left(\frac{\log n }{\epsilon^2}\right)$.
		
		To execute the exponential mechanism, we want that the grid point $g \in G_i$ be sampled with the probability $P(g)$ given by the expression 
		\begin{align*}
			P(g)= \frac{\exp\left(\frac{\epsilon_E \lvert \cover[g] \rvert}{2} \right) }{\sum_{h \in G_i} \exp\left( \frac{\epsilon_E \lvert \cover[h] \rvert}{2}  \right) }.
		\end{align*}
		Since all but polynomially many grid points $\{g : \cover[g]  = 0\}$ are being sampled with exactly the same probability, which also happens to be the smallest value any point is sampled with, we can use the law of total probability to write this sampling distribution as a uniform distribution on the entire grid with some probability $1-P_{samp}$, and a second distribution with $P'$ supported only on the polynomially many grid points with non-zero cover with probability $P_{samp}$.
		\begin{align*}
			P(g) &= P_{samp} P'(g)  + (1-P_{samp}) \frac{1}{\lvert G_i \rvert} \\
		\end{align*}
		Letting $g_0$ be any grid point with $\cover[g_0] = \emptyset$, so that $P'(g_0) = 0$,
		\begin{align*}
			(1 - P_{samp}) \frac{1}{\lvert G_i \rvert} &= \frac{1}{\sum_{h \in G_i} \exp\left( \frac{\epsilon_E \lvert \cover[h] \rvert}{2}\right)} \\
			\Rightarrow P_{samp} &= 1 - \frac{\lvert G_i \rvert}{\sum_{h \in G} \exp\left( \frac{\epsilon_E \lvert \cover[h] \rvert}{2}\right)}
		\end{align*}
		Putting together the last 3 displays, we get an expression for $P'(g)$:
		\begin{align*}
			P_{samp} P'(g) &=  \frac{\exp\left(\frac{\epsilon_E \lvert \cover[g] \rvert}{2} \right) -1}{\sum_{h \in G} \exp\left( \frac{\epsilon_E \lvert \cover[h] \rvert}{2}  \right) } \\
			\Rightarrow P'(g) &= \frac{\exp\left(\frac{\epsilon_E \lvert \cover[g] \rvert}{2} \right) -1}{\sum_{h \in G_i} \exp\left( \frac{\epsilon_E \lvert \cover[h] \rvert}{2}  \right) } \frac{\sum_{h \in G} \exp\left( \frac{\epsilon_E \lvert \cover[h] \rvert}{2}\right)}{\sum_{h \in G_i} \exp\left( \frac{\epsilon_E \lvert \cover[h] \rvert}{2}\right) - \lvert G_i \rvert }\\
			&= \frac{\exp\left(\frac{\epsilon_E \lvert \cover[g] \rvert}{2} \right) -1}{\sum_{h \in G_i} \exp\left( \frac{\epsilon_E \lvert \cover[h] \rvert}{2}\right) - \lvert G_i \rvert }
		\end{align*}
		Suppressing the normalization, the derived expression can be summarised as $P'(g) \propto \exp\left(\frac{\epsilon_E \lvert \cover[g] \rvert}{2} \right) -1$.
	\end{proof}

%% file: utility.tex
\section{Utility}
	
	The outline of the utility analysis is as follows; we know that the optimal $k$-means solution $\OPT_{D'}$ leads to $k$ clusters, each of which has some radius between $0$ and $1$. We will try to catch these clusters at threshold radii $r_1 = (1/n), r_2 = (1 + \epsilon)(1/n), r_3 = (1 + \epsilon)^2 (1/n) ,\dots, r_m = 2$ for $i=1,\dots m = \log_{1 + \epsilon} 2n$. If $o_i$ is the number of points in $D'$ such that for all $p\in o_i$, $d(p, D') \in [r_{i-1},r_i)$, then we can relate the cost of the optimal $k$-means solution as
	\begin{align*}
		\frac{\sum_{i=1}^m \lvert o_i \rvert r_i}{1 + \epsilon} \leq f_{D'} (\OPT_{D'}) < \sum_{i=1}^m \lvert o_i \rvert r_i.
	\end{align*}
	We show that the sets $a_i$, i.e. the set of points covered within $r_i + t_i \sqrt{d'}$ is close to the number of points that lie within $r_i$ of their closest mean in $\OPT_{D'}$ (\cref{lem:privateMaxCoverRelax}). It will then follow that moving each data point to its closest point in $C = \cup_{i=1}^m C_i$ will lead to moving datapoints a total distance of $\sim f_{D'} (\OPT_{D'})$ (\cref{lem:costIncrease}). Similarly it will also follow that the proxy dataset $D''$ constructed by enumerating points in $C$ with multiplicity the number of points moved to them will have a similar clustering cost (\cref{lem:D''ClusterCostBound}), and that cluster centers for $D''$ also work well as cluster centers for $D'$(\cref{lem:D'ClusterCostBound}). We identify points in $D$ with their images in $D'$ under $T'$ and show that noisy averaging of clusters in $D$ so found gives us a good solution for the $k$-means problem for $D$ (\cref{thm:finalUtility}).
	
	\begin{lemma}\label{lem:privateMaxCoverRelax}
		With probability $1-\gamma$
		\begin{align*}
			\lvert a_l \rvert \geq (1-\epsilon) \left| o_l \right| - O\left(\frac{k \log n}{\epsilon_E \cdot  \poly(\epsilon)} \log\frac{n }{\gamma}\right).
		\end{align*}
		where $\epsilon_E$ is the privacy parameter used in the exponential mechanism.
	\end{lemma}
	
	\begin{proof}
		Since the $\ell_2$ distance between $\mu$ and $\snap{\mu}{l}$ is at most $t_l\sqrt{d'}$, it follows from the definition of $o_l$ that $o_l \subset B_{r_l + t_l \sqrt{d'}} (\{\lfloor \mu \rfloor : \mu \in \OPT_{D'} \})$. We can hence apply \cref{lem:maxCoverRelax} with the promise that the set of $k$ balls with centers in $\{ \argmin_{g \in G'} d(g, c) : c \in \OPT_{D'} \}$ and radii $r_l + t_l \sqrt{d'}$ cover the set $o_l$, and that the $k$ balls lie in the family of sets $\{B_{r_l + t_l \sqrt{d'}} (g) : g \in G_l \}$.
		
		We let $h$ denote the submodular function $\lvert o_l \cap B_{t_l + r_l\sqrt{d}}(\cdot ) \rvert$ and $\Delta_i h$ denote the marginal utility function, i.e. the increase in $h$ when picking the $i$th element from the domain and adding it to the set of $i-1$ elements already picked. If $g_i^{EM}$ is the $i$th element picked by the exponential mechanism, the \cref{lem:expMech} guarantee gives us that with probability $1-\gamma$,
		\begin{align*}
		     \Delta_i h(g_i^{EM}) \geq \max\limits_{g \in G_l \backslash\{g_1^{EM},\dots,g_{i-1}^{EM} \}} \Delta_i h(g) - \frac{2}{\epsilon_E} \log \frac{\lvert G_l \rvert}{\gamma}. 
		\end{align*}
		This implies that when
		\begin{align*}  \max\limits_{g \in G_l \backslash\{g_1^{EM},\dots,g_{i-1}^{EM} \}} \Delta_i h(g) \geq \frac{4}{\epsilon_E} \log \frac{\lvert G_l \rvert}{\gamma},
		\end{align*}
		with probability $1-\gamma$,
		\begin{align}\label{eqn:marginalLowerBound}
		\Delta_i h(g_i^{EM}) \geq \frac{\max\limits_{g \in G_l \backslash\{g_1^{EM},\dots,g_{i-1}^{EM} \}} \Delta_i h(g)}{2}.
		\end{align}
		Let $g_1^{EM}, \dots , g_{k'}^{EM}$ be the grid points chosen by the exponential mechanism in the course of \cref{alg:privateGridSetCover}. Note that this implies 
		\begin{align}
			B_{r_i + t_i \sqrt{d'}} (\{ g_1^{EM}, \dots, g_{k'}^{EM}\}) \cap o_l &\subset  a_l \nonumber\\
			\Rightarrow h(\{ g_1^{EM}, \dots, g_{k'}^{EM}\}) &\leq \lvert a_l \rvert. \label{eqn:a_i}
		\end{align}
		If $j$ is the greatest index for which $\max\limits_{g \in G_l \backslash\{g_1^{EM},\dots,g_{i-1}^{EM} \}} \Delta_j h(g) \geq \frac{4}{\epsilon_E} \log \frac{\lvert G_l\rvert}{\gamma}$ (noting that the maximum possible marginal increase in cover is non-increasing), then let $g_{j+1}^{MAX}, \dots , g_{k'}^{MAX}$ be the grid points with maximal marginal utility in the round indicated by the subscript. Combining \cref{lem:maxCoverRelax} with \cref{eqn:marginalLowerBound} gives us that with probability at least $1- j \gamma \geq 1 - k' \gamma$
		\begin{align*}
			h(\{ g_1^{EM}, \dots , g_{j}^{EM} , g_{j+1}^{MAX}, \dots , g_{k'}^{MAX} \}) \geq (1-\epsilon) \lvert o_l \rvert \\
			\Rightarrow \sum_{i=1}^j \Delta_i h (g_i^{EM}) + \sum_{i=j+1}^{k'} \Delta_i h(g_i^{MAX}) \geq (1-\epsilon) \lvert o_l \rvert \\
			\Rightarrow \sum_{i=1}^j \Delta_i h (g_i^{EM}) + \sum_{i=j+1}^{k'} \frac{4}{\epsilon_E} \log \frac{|G_l|}{\gamma} \geq (1-\epsilon) \lvert o_l \rvert \\
			\Rightarrow h(\{ g_1^{EM}, \dots , g_{k'}^{EM} \}) + \frac{4 k'}{\epsilon_E} \log \frac{|G_l|}{\gamma} \geq (1-\epsilon) \lvert o_l \rvert
		\end{align*}
		We recall that $|G_l|$ is $\left\lfloor \frac{1}{t_l}\right\rfloor^{d'}$, and that $t_l$ is at least $\frac{\epsilon}{n \sqrt{d'}}$. Using this, and \cref{eqn:a_i}, we see that with probability $1-k' \gamma$,
		\begin{align*}
			\lvert a_l \rvert \geq (1-\epsilon) \left| o_l \right| - O\left(\frac{k' d'}{\epsilon_E} \log\frac{n \sqrt{d'} }{\gamma \epsilon}\right).
		\end{align*}
		We absorb the $k' = \frac{k}{\epsilon}$ factor into the failure probability $\gamma$ and noting that $d' = O\left(\frac{\log n}{\epsilon^2}\right)$, $k' = \frac{k}{\epsilon}$ and $k<n$ gives us
		\begin{align*}
			\lvert a_l \rvert &\geq (1-\epsilon) \left| o_l \right| - O\left(\frac{k \log n}{\epsilon^3 \cdot \epsilon_E} \log\frac{k n \sqrt{\log n} }{\gamma \epsilon^3}\right) \\
			&\geq (1-\epsilon) \left| o_l \right| - O\left(\frac{k \log n}{\epsilon_E \cdot \poly(\epsilon)} \log\frac{n }{\gamma}\right).
		\end{align*}
	\end{proof}
	
	We see that although the cluster radii thresholds are $r_i$ for $i=1\dots m$, discretization leads to slightly inflated cluster radii $r_i + t_i \sqrt{d} = (1 + \epsilon)r_i$. This also leads to a slight multiplicative inflation in the total movement, showing up as the $(1+\epsilon)$ coefficient of $\sum_{i=1}^m \lvert a_i \rvert r_i$ in the following result bounding the total movement of points when constructing the proxy dataset $D''$ (without addition of noise).
	
	\begin{lemma}\label{lem:costIncrease}
		The total movement of points $p \in D'$ to the closest point $\grid[p] \in C$ is bounded by the following inequalities
		\begin{align*}
			\sum_{p \in D'} d(p, \grid[p]) \leq (1+ \epsilon) \sum_{i=1}^m |a_i| r_i \leq\left(1 + \frac{3\epsilon}{1 - \epsilon - \epsilon^2}\right) f_{D'} (\OPT_{D'}) + O\left(\frac{k \log n}{\epsilon_E\cdot \poly(\epsilon)} \log \frac{n}{\gamma} \right).
		\end{align*}
	\end{lemma}
	
	We break this proof down into a couple of smaller steps.

	\begin{lemma}
		The thresholded cost obeys the bound
		\begin{align}\label{eqn:step1}
			\sum_{i=1}^m |o_i| r_i \leq (1 + \epsilon) f_{D'}(\OPT_{D'}) + 1.
		\end{align}
	\end{lemma}
	
	\begin{proof}
		We recall that $o_i = \{p \in D' : d(p, \OPT_{D'}) \in [r_{i-1}, r_i)\}$. Since $D' \in B_1 (0)$, and $r_m = 2$, $\cup_{i=1}^m a_i = D'$. From this we have that
		\begin{align*}
			f(\OPT_{D'}) &= \sum_{i=1}^m \sum_{p \in o_i} d(p,\OPT_{D'}) \\
			&= \sum_{p \in o_1} d(p,\OPT) + \sum_{i=2}^m \sum_{p \in o_i} d(p, \OPT_{D'}).
		\end{align*}
		We note that for $d(p,\OPT_{D'}) \in [0,r_1) = [0,1/n) \Rightarrow \frac{1}{n} + d(p,\OPT_{D'}) > r_1$, since $r_1 = 1/n$ so 
		\begin{align*}
			\sum_{p\in o_1} \frac{1}{n} + d(p,\OPT_{D'}) &> \sum_{p \in o_1 } r_1 \\
			\Rightarrow 1 + \sum_{p \in o_1} d(p,\OPT_{D'}) &\geq |o_1| r_1.
		\end{align*}
		For $d(p,\OPT_{D'}) \in [r_{i-1},r_i)$ for $i \not= 1$, since $\frac{r_i}{r_{i-1}} = (1 + \epsilon)$, it follows that $d(p,\OPT_{D'}) > \frac{r_i}{1 + \epsilon}$, so summing over all such $p$ we have
		$$ \sum_{d(p,\OPT_{D'}) \in [t_{i-1},t_i)} d(p,\OPT_{D'}) > \frac{|o_i| r_i}{1 + \epsilon}.$$
		From the last two displays we have 
		\begin{align*}
			\sum_{i=1}^m |o_i| r_i \leq (1 + \epsilon) f_{D'}(\OPT_{D'}) + 1.
		\end{align*}
		
	\end{proof}
	
	\begin{lemma}
		The ideal thresholded cost $\sum_{i=1}^m \lvert o_i \rvert r_i$ can be related to the achieved thresholded cost $\sum_{i=1}^m \lvert a_i \rvert r_i$ by the following inequality
		\begin{align}
			\sum_{i=1}^m \lvert a_i\rvert r_i &\leq \frac{1-\epsilon}{1 - \epsilon - \epsilon^2} \sum_{i=1}^m \left(\lvert o_i\rvert r_i \right) + O\left(\frac{k \log n}{\epsilon_E \cdot \poly(\epsilon)} \log\frac{n }{\gamma}\right). \label{eqn:step2}
		\end{align}
	\end{lemma}
	
	\begin{proof}
		We define $O_i = \sum_{j=i}^m |o_j|$ and $A_i = \sum_{j=i}^m |a_j|$. We then have
		\begin{align}\label{eqn:telescope}
			\sum_{i=1}^m \lvert a_i \rvert r_i &= \sum_{i=1}^m A_i (r_i - r_{i-1}).
		\end{align}
		We note that centers in $\OPT_{D'}$ cover $n - O_{i+1}$ points at a maximum distance of $r_i$. We also know that \cref{alg:privateKMeans} has already covered $n - A_i$ points at a distance of $r_{i-1} + t_{i-1}\sqrt{d'}$. It then follows that there are some $k$ grid points in $G_i$ (snapping the centers in $\OPT_{D'}$ to grid) that cover at least $(n - O_{i+1}) - (n - A_i) = A_i - O_{i+1}$ uncovered points at a distance of at a distance of $r_i + t_i \sqrt{d'}$. From the \cref{lem:privateMaxCoverRelax} guarantee, we know that
		\begin{align*}
			\lvert a_i \rvert \geq (1-\epsilon)(A_i - O_{i+1}) - E,
		\end{align*}
		where $E = O\left(\frac{k \log n}{\epsilon_E \cdot \poly(\epsilon)} \log\frac{n }{\gamma}\right)$. Since $A_i = \lvert a_i \rvert + A_{i+1}$, we have that
		\begin{align*}
			\lvert a_i \rvert \geq \left(\frac{1-\epsilon}{\epsilon}\right)(A_{i+1} - O_{i+1}) - \frac{E}{\epsilon} \\
			\Rightarrow A_{i+1} \leq \left(\frac{\epsilon}{1-\epsilon}\right) \lvert a_i \rvert + O_{i+1} + \frac{E}{1-\epsilon}.
		\end{align*}
		Substituting this in \cref{eqn:telescope}, we continue as follows:
		\begin{align*}
			\sum_{i=1}^m \lvert a_i\rvert r_i &=\sum_{i=1}^m A_i (r_i - r_{i-1}) \\
			&\leq \sum_{i=1}^m \left( \left(\frac{\epsilon}{1-\epsilon}\right) \lvert a_{i-1} \rvert + O_{i} + \frac{E}{1-\epsilon} \right) (r_i - r_{i-1}) \nonumber \\
			&= \sum_{i=1}^m \left(\frac{\epsilon \lvert a_{i-1} \rvert}{1-\epsilon}\right) (r_i - r_{i-1}) + \sum_{i=1}^m \lvert o_i\rvert r_i + \frac{E}{1-\epsilon} (r_m - r_0) \\
			&\leq \sum_{i=1}^m\left(\frac{\epsilon^2 \lvert a_{i-1}\rvert}{1-\epsilon}\right)  r_{i-1} + \sum_{i=1}^m \lvert o_i\rvert r_i + \frac{2E}{1-\epsilon}\\
			\Rightarrow \sum_{i=1}^m \lvert a_i\rvert r_i \left(1 - \frac{\epsilon^2}{1-\epsilon} \right) &\leq \sum_{i=1}^m \lvert o_i\rvert r_i + \frac{2E}{1-\epsilon} \\
			\Rightarrow \sum_{i=1}^m \lvert a_i\rvert r_i &\leq \frac{1-\epsilon}{1 - \epsilon - \epsilon^2} \sum_{i=1}^m \lvert o_i\rvert r_i + \frac{2E}{1 - \epsilon - \epsilon^2} 
		\end{align*}
		Substituting the order term $E$ and using that $\epsilon$ is bounded away from $1$, we get the desired inequality.
	\end{proof}
	
	\begin{proof}[Proof of \cref{lem:costIncrease}]
		We recall that $a_i = D' \cap B_{r_i + t_i \sqrt{d'}}(C_i)$. Since $D' \in B_1 (0)$, and $r_m = 2$, $\cup_{i=1}^m a_i = D'$. It follows that
		\begin{align*}
			\sum_{p \in D'} d(p, \grid[p]) &= \sum_{i=1}^m \sum_{ p \in a_i } d(p, \grid[p]) \\
			&\leq \sum_{i=1}^m \sum_{ p \in a_i } r_i + t_i \sqrt{d} \\
			&\leq \sum_{i=1}^m \lvert a_i \rvert (r_i + t_i \sqrt{d}) \\
			&\leq (1 + \epsilon) \sum_{i=1}^m \lvert a_i \rvert r_i
		\end{align*}
		This proves the first inequality. 
		From \cref{eqn:step1} and \cref{eqn:step2} we see that $\sum_{i=1}^m |a_i| r_i$ obeys the inequality
		\begin{align}
			\sum_{i=1}^m |a_i| r_i &\leq \frac{1-\epsilon}{1 - \epsilon - \epsilon^2} ( ( 1 + \epsilon) f_{D'}(\OPT_{D'}) + 1) + O\left(\frac{k \log n}{\epsilon_E \cdot \poly(\epsilon)} \log\frac{n }{\gamma}\right) \nonumber \\
			\Rightarrow (1 + \epsilon) \sum_{i=1}^m |a_i| r_i &\leq \left(1 + \frac{\epsilon}{1 - \epsilon - \epsilon^2}\right) ( ( 1 + \epsilon) f_{D'}(\OPT_{D'}) + 1) +O\left(\frac{k \log n}{\epsilon_E\cdot \poly(\epsilon)} \log\frac{n }{\gamma}\right)
		\end{align}
		with probability $1-\gamma$. Absorbing smaller order terms into the additive error and using \cref{lem:costIncrease}, we get
		\begin{align*}
			\sum_{i=1}^m d(p, \grid[p]) \leq \left(1 + \frac{3\epsilon}{1 - \epsilon - \epsilon^2}\right) f_{D'} (\OPT_{D'}) + O\left(\frac{k \log n}{\epsilon_E\cdot \poly(\epsilon)} \log \frac{n}{\gamma} \right).
		\end{align*}	
	\end{proof}
	
	In \cref{lem:D''ClusterCostBound} we bound the $k$-means cost of the proxy dataset $D''$ in terms of the $k$-means cost of $D'$. 
	
	\begin{lemma}\label{lem:D''ClusterCostBound}
		With probability $1-\gamma$,
		\begin{align*}
			f_{D''} (\OPT_{D'}) \leq \left(4 + \frac{6\epsilon}{1 - \epsilon - \epsilon^2} \right) f_{D'} (\OPT_{D'}) + O\left(\frac{k \log n}{\epsilon_E\cdot \poly(\epsilon)} \log \frac{n}{\gamma} \right) + O\left(\frac{k \log n}{\epsilon_L \cdot \poly(\epsilon)}\right).
		\end{align*}
	\end{lemma}
	
	Since $D''$ is constructed by using noisy counts of candidate centers, we defer its proof and first derive a technical lemma to bound the error added by the privacy preserving algorithm.
	
	\begin{lemma} \label{lem:chernoff}
		Let $X_1, \dots X_{k'}$ be i.i.d. $\Lap (\frac{1}{\epsilon_L})$ random variables. Then for any positive constants $c_1, \dots , c_{k'}$, with probability $1-\gamma$,
		\begin{align*}
			X:= \sum_{i=1}^{k'} c_i \lvert X_i \rvert \leq (\log 2) \sum_{i=1}^{k'} \frac{c_i}{ \epsilon_L} + \frac{2 \max_i c_i}{\epsilon_L} \log 1/\gamma .
		\end{align*}
	\end{lemma}
	
	\begin{proof}
		The moment generating function of $c \lvert X \rvert$ for $X \sim \Lap(\frac{1}{\epsilon_L})$ is $\frac{\epsilon_L}{\epsilon_L - c t} = \frac{1}{1 - \frac{ct}{\epsilon_L}}$ for $\lvert c t \rvert < \epsilon_L$. It then follows that for $t \leq \frac{\epsilon_L}{\max_i c_i}$,
		\begin{align*}
			M_X (t) = \prod_{i=1}^{k'} \frac{1}{1 - \frac{c_i t}{\epsilon_L}} .
		\end{align*}
		We use the Chernoff bound
		\begin{align*}
			P(X > \Delta n) &= P(e^{tX} > e^{t \Delta n}) \\
			&\leq \frac{\mathbb{E}[e^{tX}]}{e^{t \Delta n}} \\
			&\leq \frac{M_X (t)}{e^{t \Delta n}}.
		\end{align*}
		If we require that this event occur with probability at most $\gamma$, we derive a bound on $\Delta n$ that would suffice.
		\begin{align*}
			\prod_{i=1}^{k'} \left(\frac{1}{1 - \frac{c_i t}{\epsilon_L}}\right) \cdot \frac{1}{e^{t \Delta n}} &\leq \gamma \\
			\Leftrightarrow \sum_{i=1}^{k'} \log \left(\frac{1}{1 - \frac{c_i t}{\epsilon_L}}\right) - t \Delta n &\leq \log \gamma \\
			\Leftrightarrow t \Delta n &\geq -\sum_{i=1}^{k'} \log \left(1 - \frac{c_i t}{\epsilon_L}\right) + \log 1/\gamma
		\end{align*}
		For $\frac{c_i t}{\epsilon_L}<0.5$, $- \log \left( 1 - \frac{c_i t}{\epsilon_L}\right) < (\log 2) \frac{c_i t}{\epsilon_L}$ by convexity of $-\log (1 - x)$ in the argument $x$. Restricting $t$ to $[0, \frac{\epsilon_L}{2\max_i c_i}]$, we continue.
		\begin{align*}
			\Leftarrow t \Delta n &\geq \log 2 \sum_{i=1}^{k'} \frac{c_i t }{ \epsilon_L} + \log 1/\gamma\\
			\Leftrightarrow \Delta n &\geq \log 2 \sum_{i=1}^{k'} \frac{c_i}{ \epsilon_L} + \frac{1}{t}\log 1/\gamma
		\end{align*}
		Setting $t = \frac{\epsilon_L}{2\max_i c_i}$, we get the desired result.
	\end{proof}
	
	\begin{proof}[Proof of \cref{lem:D''ClusterCostBound}]
		In $D''$ each grid point occurs with multiplicity $n_c' = n_c + X_c$ where $X_c \sim \mbox{Lap}\left(\frac{2}{\epsilon_L} \right)$.
		\begin{align*}
			f_{D''} (\OPT_{D'}) &= \sum_{p \in D''} d(p, \OPT_{D'}) \\
			&= \sum_{c \in C} n_c' d(c, \OPT_{D'}) \\
			&= \sum_{c \in C} \left(n_c d(c, \argmin_{\mu \in D''} d(c,\mu)) + X_c d(c, \argmin_{\mu \in D''} d(c,\mu))\right) \\
			&\leq \sum_{c \in C} n_c d(c, \argmin_{\mu \in \OPT_{D'}} d(c,\mu)) + \sum_{c \in C} \lvert X_c \rvert \cdot 2
		\end{align*}
		Using that points in $D''$ are enumerated by running over centers in $c \in C$ with multiplicity $n_c$, and by applying \cref{lem:chernoff}, we get that with probability $1-\gamma$,
		\begin{align*}
			f_{D''} (\OPT_{D'}) &\leq \sum_{p \in D'} d(\grid[p], \argmin_{\mu \in \OPT_{D'}} d(\grid[p],\mu)) + O\left(\frac{k \log n}{\epsilon^2\cdot \epsilon_L}\right) + O \left( \frac{\log 1/\gamma}{\epsilon_L} \right)\\
			&\leq \sum_{p \in D'} 2 \left(d(\grid[p], p) + d(p, \argmin_{\mu \in \OPT_{D'}} d(\grid[p],\mu))\right) + O\left(\frac{k \log n}{\epsilon_L \cdot \poly(\epsilon)}\right).
		\end{align*}
		In the above we drop $\log 1/\gamma$ as it is asymptotically dominated by the other error term for any failure probability polynomially small in $n$. Using \cref{lem:costIncrease} to simplify the first term, we have that with probability $1-2\gamma$,
		\begin{align*}
			f_{D''} (\OPT_{D'}) &\leq 2 \left(1 + \frac{3\epsilon}{1 - \epsilon - \epsilon^2}\right) f_{D'} (\OPT_{D'}) + 2 f_{D'}(\OPT_{D'}) + O\left(\frac{k \log n}{\epsilon_E\cdot \poly(\epsilon)} \log \frac{n}{\gamma} \right) + O\left(\frac{k \log n}{\epsilon_L \cdot \poly(\epsilon)}\right) \\
			&\leq \left(4 + \frac{6\epsilon}{1 - \epsilon - \epsilon^2} \right) f_{D'} (\OPT_{D'}) + O\left(\frac{k \log n}{\epsilon_E\cdot \poly(\epsilon)} \log \frac{n}{\gamma} \right) + O\left(\frac{k \log n}{\epsilon_L \cdot \poly(\epsilon)}\right).
		\end{align*}
	\end{proof}
	
	In \cref{lem:D'ClusterCostBound} we bound the error incurred on using the $k$-means solution for the proxy dataset $D''$ for the dataset $D'$.
	
	\begin{lemma}\label{lem:D'ClusterCostBound}
		Let $\mathcal{A}$ be the clustering algorithm used in \cref{alg:privateKMeans;line:cluster} of \cref{alg:privateKMeans}. If $\mathcal{A}$ has the utility guarantee
		\begin{align*}
			f_S(\mathcal{A}(S)) \leq E_M \cdot f_{S}(\OPT_S) + E_A
		\end{align*}
		then 
		\begin{align*}
			f_{D'} (\mathcal{A}(D'')) \leq (8 E_M + 2 + (8 E_M + 4) \epsilon) f_{D'}(\OPT_{D'}) + 2 E_A + O\left(\frac{k \log n}{\epsilon_E} \log \frac{kn}{\gamma}\right).
		\end{align*}
	\end{lemma}
	
	\begin{proof}
		From the clustering algorithm guarantee we have that
		\begin{align*}
			f_{D''}(\mathcal{A}(D'')) \leq E_M \cdot f_{D''}(\OPT_{D''}) + E_A.
		\end{align*}
		By definition, $f_{D''} (\OPT_{D''}) < f_{D''}(\OPT_{D'})$. Substituting the bound from \cref{lem:D''ClusterCostBound}, we get
		\begin{align*}
			f_{D''}(\mathcal{A}(D'')) \leq \left(8 + \frac{12\epsilon}{1 - \epsilon - \epsilon^2} \right) (E_M + 1)  f_{D'}(\OPT_{D'}) + 2 E_A + O\left(\frac{k \log n}{\epsilon_E\cdot \poly(\epsilon)} \log \frac{n}{\gamma} \right) + O\left(\frac{k \log n}{\epsilon_L \cdot \poly(\epsilon)}\right).
		\end{align*}
		So then
		\begin{align*}
			f_{D'} ( \mathcal{A}(D'')) &= \sum_{p \in D'} d(p, \argmin_{\mu \in \mathcal{A}(D'')} d(p,\mu)) \\
			&\leq \sum_{p \in D'} d(p, \argmin_{\mu \in \mathcal{A}(D'')} d(\grid[p],\mu)) \\
			&\leq \sum_{p \in D'} 2 \left(d(p, \grid[p]) + d(\grid[p], \argmin_{\mu \in \mathcal{A}(D'')} d(\grid[p],\mu))\right) \\
			&\leq \left(2 \sum_{p \in D'} d(p, \grid[p]) \right) + 2 f_{D''}(\mathcal{A}(D'')) \\
			&\leq  2 \cdot \left(4 + \frac{6\epsilon}{1 - \epsilon - \epsilon^2} \right) f_{D'} (\OPT_{D'}) + O\left(\frac{k \log n}{\epsilon_E\cdot \poly(\epsilon)} \log \frac{n}{\gamma} \right) + O\left(\frac{k \log n}{\epsilon_L \cdot \poly(\epsilon)}\right)  + \\
			&2 E_M \cdot \left(4 + \frac{6\epsilon}{1 - \epsilon - \epsilon^2} \right) f_{D'} (\OPT_{D'}) + O\left(\frac{k \log n}{\epsilon_E\cdot \poly(\epsilon)} \log \frac{n}{\gamma} \right) + O\left(\frac{k \log n}{\epsilon_L \cdot \poly(\epsilon)}\right)+ 2 E_A \\
			&\leq \left(8 + \frac{12\epsilon}{1 - \epsilon - \epsilon^2} \right) (E_M + 1)  f_{D'}(\OPT_{D'}) + 2 E_A + O\left(\frac{k \log n}{\epsilon_E\cdot \poly(\epsilon)} \log \frac{n}{\gamma} \right) + O\left(\frac{k \log n}{\epsilon_L \cdot \poly(\epsilon)}\right).
		\end{align*}
	\end{proof}
	To complete the utility analysis, we need to account for the projection and scaling as well as the Gaussian noise added to maintain privacy. In \cref{lem:projectionAndScaling} we derive an expression for the utility without accounting for any noise in and \cref{thm:finalUtility} we derive an expression for the net utility guarantee of \cref{alg:privateKMeans}. 
	
	\begin{lemma}\label{lem:projectionAndScaling}
		If we cluster $D$ according to its projected and scaled version $D'$, we get a set $S$ of size $k$ such that
		\begin{align*}
			 f_D (S) &\leq (1 + \epsilon) \left(8 + \frac{12\epsilon}{1 - \epsilon - \epsilon^2} \right) (E_M + 1)  f_{D}(\OPT_{D}) + 2(1 + \epsilon)\Delta^2 E_A + \\
			 & + O\left(\frac{k \Delta^2 \log n}{\epsilon_E\cdot \poly(\epsilon)} \log \frac{n}{\gamma} \right) + O\left(\frac{k \Delta^2 \log n}{\epsilon_L \cdot \poly(\epsilon)}\right).
		\end{align*}
	\end{lemma}
	
	\begin{proof}
		To scale the data from a ball that lies within a diameter of $\Delta$ to a diameter of $2$, we note that the clustering cost is multiplied by a factor of $\Delta^2/4$. To account for the projection we recall that the $k$-means cost can be expressed without explicit reference to the means themselves. We let $\{ D'_i : i \in [k] \}$ be the partition of the data $D'$ into $k$ clusters, where the $i$th cluster is centered at $\mu'_i$ and $S' = \{\mu'_i : i \in [k]\}$.
		\begin{align*} 
			f_{D'} (S') &= \sum_{i \in [k]} \sum_{p \in D'_i} d(p, \mu'_i) \\
			&= \sum_{i \in [k]} \sum_{p \in D'_i} \lVert p - \mu'_i \rVert^2 \\
			&= \sum_{i \in [k]} \frac{1}{\lvert D'_i \rvert} \sum_{p\not= q \in D'_i} \lVert p - q\rVert^2.
		\end{align*}
		Since the Johnson Lindenstrauss transform preserves the $\ell_2$ norm squared within a multiplicative factor of $(1 \pm \epsilon)$, it follows from the display above that the cost of clustering $D$ according to its image $D'$ is at most $(1 + \epsilon) f_{D'}(S')$. Denoting the cluster centers derived in this fashion by $S$, this gives us
		\begin{align*}
			f_D (S) &\leq (1 + \epsilon)\left(2 + \frac{4\epsilon}{1 - \epsilon - \epsilon^2} \right) (E_M + 1)  \Delta^2 f_{D'}(\OPT_{D'}) +  \frac{\Delta^2 E_A}{2} + \\
			& O\left(\frac{k \Delta^2 \log n}{\epsilon_E\cdot \poly(\epsilon)} \log \frac{n}{\gamma} \right) + O\left(\frac{k \Delta^2 \log n}{\epsilon_L \cdot \poly(\epsilon)}\right) \\
			\Rightarrow f_D (S) &\leq (1 + \epsilon) \left(8 + \frac{12\epsilon}{1 - \epsilon - \epsilon^2} \right) (E_M + 1)  f_{D}(\OPT_{D}) + 2(1 + \epsilon)\Delta^2 E_A + \\
			& + O\left(\frac{k \Delta^2 \log n}{\epsilon_E\cdot \poly(\epsilon)} \log \frac{n}{\gamma} \right) + O\left(\frac{k \Delta^2 \log n}{\epsilon_L \cdot \poly(\epsilon)}\right),
		\end{align*}
		where we use that $\frac{\Delta^2}{4} f_{D'} (\OPT_{D'}) = f_D (\OPT_D)$.
	\end{proof}
	
	\begin{theorem}\label{thm:finalUtility}
		\Cref{alg:privateKMeans} returns a set of points $\tilde{S}$ such that
		\begin{align*}
			\mathbb{E}\left[f_D (\tilde{S})\right] &\leq (1 + \epsilon) \left(8 + \frac{12\epsilon}{1 - \epsilon - \epsilon^2} \right) (E_M + 1)  f_{D}(\OPT_{D}) + 2(1 + \epsilon)\Delta^2 E_A + \\
			& + O\left(\frac{k \Delta^2 \log n}{\epsilon_E\cdot \poly(\epsilon)} \log \frac{n}{\gamma} \right) + O\left(\frac{k \Delta^2 \log n}{\epsilon_L \cdot \poly(\epsilon)}\right) + O \left( \frac{k \Delta^2 \sqrt{d \log 1/\delta_G }}{\epsilon_G} \right) + O \left( \frac{k \Delta^2 \log n/\delta_G }{\epsilon_G} \right).
		\end{align*}
	\end{theorem}
	
	\begin{proof}
		The final set of points returned, denoted $\tilde{S}$, is obtained by using \cref{alg:noisyAVG}. From the statement of \cref{thm:noisyAVG}, we know that for the $i$th cluster if $\lvert D_i \rvert \ge A\left( \frac{1}{\epsilon_G} \log \left( \frac{nk}{\delta_G} \right) \right)$ with sufficiently large constant $A$ then with probability $1- \frac{1}{k n}$, \cref{alg:noisyAVG} returns $\mu_i + g_i$ where $g_i$ is sampled from $N(0,\sigma^2)$ for some $\sigma < \frac{4 \Delta}{\epsilon_G \lvert D_i \rvert} \sqrt{2 \ln (4/\delta_G)}$. We let $c = 4\sqrt{2 \ln (4/\delta_G)}$ so that $\sigma<\frac{c \Delta}{ \epsilon_G \lvert D_i \rvert}$. We can upper bound the clustering cost by assuming that cluster sets remain the same even with the noisy means, and then add up the cost cluster by cluster.
		\begin{align*}
			f_D (\tilde{S}) &= \sum_{p \in D} d(p, \tilde{S}) \\
			&\leq \sum_{i \in [k]} \sum_{p \in D_i} d(p, \tilde{\mu_i}) \\
			\sum_{p \in D_i} d(p, \tilde{\mu_i})  &= \sum_{p \in D_i} \lVert x - \tilde{\mu_i} \rVert^2\\
			&= \sum_{p \in D_i} \lVert p - \mu_i + g_i \rVert^2 \\
			&= \sum_{p \in D_i} \langle p - \mu_i + g_i, p - \mu_i + g_i \rangle \\
			&= \left(\sum_{p \in D_i} \lVert p - \mu_i \rVert^2 + \left\langle \sum_{p \in D_i} p- \mu_i, g_i \right\rangle + \sum_{p \in D_i} \lVert g_i \rVert^2 \right) \\
			&= f_{D_i} (\{ \mu_i \}) + \left\langle(|D_i| - 1)\mu_i, g_i\right\rangle + \sum_i \lvert D_i \rvert \lVert g_i \rVert^2,
		\end{align*}
		where in the last step we use that $\sum_{p \in D_i} p = \lvert D_i \rvert \mu$. If $\lvert D_i \rvert \ge A \left( \frac{1}{\epsilon_G} \log \left( \frac{nk}{\delta_G} \right)\right)$ for sufficiently large constant $A$, then taking the expectation, we get
		\begin{align*}
			\Rightarrow \mathbb{E}\left[\sum_{p \in D_i} d(p, \tilde{S})\right] &\leq f_{D_i} (\{ \mu_i \})  + \lvert D_i \rvert \mathbb{E}\left[\sum_{j=1}^d\lVert g_i \rVert^2 \right] \\
			&\leq f_{D_i} (\{ \mu_i \}) + \lvert D_i \rvert \left( \frac{c \Delta}{\lvert D_i \rvert \epsilon_G} \right)^2  d \\
			&\leq f_{D_i} (\{ \mu_i \}) + \frac{c^2 \Delta^2}{\lvert D_i \rvert \epsilon_G^2}   d.
		\end{align*}
		If $\lvert D_i \rvert \geq \frac{c \sqrt{d}}{\epsilon_G}$ then this is at most $f_{D_i} (\{ \mu_i\}) + \frac{c \Delta^2}{\epsilon_G} \sqrt{d}$. On the other hand, if $\lvert D_i \rvert < \frac{c \sqrt{d}}{\epsilon_G}$, we observe that the clustering cost $f_{D_i} (\tilde{\mu}_i)$ can be at most $ \frac{c \sqrt{d}}{\epsilon_G} \Delta^2$ unconditionally. Similarly if $\lvert D_i \rvert = O \left( \frac{1}{\epsilon_G} \log \left( \frac{nk}{\delta_G}\right) \right)$, then the clustering cost $f_{D_i} (\tilde{\mu}_i)$ can be at most $ O\left( \frac{1}{\epsilon_G} \log \left( \frac{nk}{\delta_G} \right) \Delta^2\right)$. With probability $1 - \frac{1}{n}$ the large cluster cost bound holds for all clusters simultaneously and we then have
		\begin{align*}
			\mathbb{E}\left[\sum_{p \in D} d(p, \tilde{S})\right] &\leq \sum_{i \in [k]} \mathbb{E}\left[\sum_{p \in D_i} d(p, \tilde{\mu_i})\right] \\
			&\leq \left( 1- \frac{1}{n} \right)\left( \sum_{i: \lvert D_i \rvert \geq \sqrt{d}} f_{D_i} (\{ \mu_i \}) + \frac{c \Delta^2}{\epsilon_G} \sqrt{d} \right) + \frac{1}{n} \Delta^2 n \\ 
			& + \sum_{i: \lvert D_i \rvert < \sqrt{d}} \frac{c \sqrt{d}}{\epsilon_G} \Delta^2 + \sum_{i: \lvert D_i \rvert < \frac{8}{\epsilon_G} \log \left( \frac{2nk}{\delta_G} \right)} O\left( \frac{1}{\epsilon_G} \log \left( \frac{nk}{\delta_G} \right) \Delta^2\right) \\
			&\leq \left(\sum_{i \in k} f_{D_i} (\{ \mu_i \}) + k \frac{c \Delta^2}{\epsilon_G} \sqrt{d} \right) + \Delta^2 + k \Delta^2  \frac{c \sqrt{d}}{\epsilon_G} + O\left( \frac{k}{\epsilon_G} \log \left( \frac{nk}{\delta_G} \right) \Delta^2\right)  \\
			&=  f_D (S) + O \left( \frac{k \Delta^2 \sqrt{d \log 1/\delta_G }}{\epsilon_G} \right) + O \left( \frac{k \Delta^2 \log n/\delta_G }{\epsilon_G} \right).
		\end{align*}
		Substituting the bound on $f_D(S)$ from \cref{lem:projectionAndScaling}, we get
		\begin{align*}
			\mathbb{E}\left[f_D (\tilde{S})\right] &\leq (1 + \epsilon) \left(8 + \frac{12\epsilon}{1 - \epsilon - \epsilon^2} \right) (E_M + 1)  f_{D}(\OPT_{D}) + 2(1 + \epsilon)\Delta^2 E_A + \\
			& + O\left(\frac{k \Delta^2 \log n}{\epsilon_E\cdot \poly(\epsilon)} \log \frac{n}{\gamma} \right) + O\left(\frac{k \Delta^2 \log n}{\epsilon_L \cdot \poly(\epsilon)}\right) + O \left( \frac{k \Delta^2 \sqrt{d \log 1/\delta_G }}{\epsilon_G} \right) + O \left( \frac{k \Delta^2 \log n/\delta_G }{\epsilon_G} \right).
		\end{align*}
	\end{proof}

%% file: privacy.tex
\section{Privacy}

	The main result of this section is the following:
	
	\begin{theorem}\label{thm:finalPrivacy}
		\Cref{alg:privateKMeans} is $\left(\frac{e \epsilon_E \ln \delta_E^{-1} }{2}+ \epsilon_L + \epsilon_G, \delta_E + \delta_G \right)$-differentially private.
	\end{theorem}

	From the basic (\cref{thm:basicComp}) and parallel (\cref{thm:parallelComp}) composition laws of differential privacy and the privacy guarantees of the Laplace mechanism (\cref{lem:lapMech}) and \cref{alg:noisyAVG} (\cref{thm:noisyAVG}) most of the expression for the bound on privacy loss claimed in this result follows relatively straightforwardly. To bound the privacy loss incurred in the calls to \cref{alg:privateGridSetCover}, we adapt a technique from \cite{gupta2010differentially}. We use this technique in the following lemma to show that the privacy loss when using the exponential mechanism many times successively can be bounded as an expression of the sum of expected gains in the cover. For the set cover function this sum of expected gains can be shown to decay exponentially using \cref{lem:gupta}, which leads to a strong bound on the privacy loss.
	
	\begin{lemma} \label{lem:loopPrivacyLoss}
		The subroutine \cref{alg:privateKMeans;line:loopStart}-\cref{alg:privateKMeans;line:loopEnd} of \cref{alg:privateKMeans} that constructs set of centers $C$ (over $m$ iterations) is $\left(\frac{e \epsilon_E \ln \delta_E^{-1} }{2},\delta_E \right)$-differentially private
	\end{lemma}
	
	\begin{proof}
		Let $A$ and $B$ be two neighbouring datasets, i.e. $A \triangle B = \{I\}$. To show that this subroutine (denoted $\mathcal{A}$) is $(\epsilon, \delta)$ differentially private, we need to show that the ratio $P(\mathcal{A}(A) = C)/P(\mathcal{A}(B) = C)$ is bounded from above by $e^\epsilon$ with probability $1-\delta$, where $C$ is an arbitrary sequence of grid points $c_1, \dots , c_{km/\epsilon}$ that might be picked in the thresholded max-cover subroutine.
		\begin{align*}
			\frac{P(\mathcal{A}(A)=C)}{P(\mathcal{A}(B)=C)} &= \prod_{i=1}^{km/\epsilon} \frac{P(\mathcal{A}(A)_i= c_i | c_1 , \dots , c_{i-1})}{P(\mathcal{A}(B)_i= c_i | c_1 , \dots , c_{i-1})} \\
			P(\mathcal{A}(A)_i= c_i | c_1 , \dots , c_{i-1}) &= \frac{\exp\left( \frac{\epsilon_E \lvert \cover [c_i] \rvert}{2} \right)}{\sum_{g} \exp\left( \frac{\epsilon_E \lvert \cover [g]\rvert}{2} \right)}\\
			\Rightarrow \frac{P(\mathcal{A}(A)=C)}{P(\mathcal{A}(B)=C)} &= \prod_{i=1}^{km/\epsilon} \frac{\exp\left( \frac{\epsilon_E \lvert \cover_A[c_i] \rvert}{2} \right)}{\exp\left( \frac{\epsilon_E \lvert \cover_B[c_i] \rvert}{2} \right)} \cdot \prod_{i=1}^{km/\epsilon} \frac{\sum_{g} \exp\left( \frac{\epsilon_E \lvert \cover_B[g] \rvert}{2} \right)}{\sum_{g} \exp\left( \frac{\epsilon_E \lvert \cover_A[g] \rvert}{2} \right)}.
		\end{align*}
		If $A \backslash B = \{ I \}$ then we see that the second factor is at most $1$ and the first factor is at most $\exp \left(\frac{\epsilon_E}{2}\right)$, since $\cover_A [c_i] \backslash \cover_B [c_i]$ can be at most the data point $I$, and that too for at most one index $i$, since $\cover$ counts only yet uncovered data points.
		Inversely if $B \backslash A = \{ I \}$, then the first factor is at most $1$ and we need to bound the second factor. We observe that this ratio of sums can be written as an expectation by factoring out the indicator of $I$ as follows:
		\begin{align*}
			\prod_{i=1}^{km/\epsilon} \frac{\sum_{g} \exp\left( \frac{\epsilon_E \lvert \cover_B[g] \rvert}{2} \right)}{\sum_{g} \exp\left( \frac{\epsilon_E \lvert \cover_A[g] \rvert}{2} \right)} &= \prod_{i=1}^{km/\epsilon} \mathbb{E}_{g \sim \exp\left(\frac{\epsilon_E \lvert \cover_A[g] \rvert}{2} \right)} \left[ \exp\left( \frac{ \epsilon_E 1_{I \in \cover_B[g] } }{2} \right) \right] \\
			&\leq \prod_{i=1}^{km/\epsilon} \mathbb{E}_{g \sim \cdot } \left[ 1 + e\cdot  \frac{\epsilon_E 1_{I \in \cover_B[g] }}{2} \right] \\
			&= \prod_{i=1}^{km/\epsilon} 1 + \frac{e \epsilon_E\mathbb{E}[1_{I \in \cover_B[g] }]}{2}  \\
			&\leq \prod_{i=1}^{km/\epsilon} \exp \left( \frac{e \epsilon_E\mathbb{E}[1_{I \in \cover_B[g] }]}{2} \right) \\
			&= \exp \left( \frac{e \epsilon_E \sum\limits_{i=1}^{km/\epsilon} \mathbb{E}[1_{I \in \cover_B[g] }]}{2} \right).
		\end{align*}
		To bound the sum of expectations that occurs in the exponent, we use \cref{lem:gupta} with $R_i = 1_{I \in \cover_B [c_i]}$ and $p_i = \mathbb{E}[R_i]$ if $I$ has not been picked by the $(i-1)$th round and $R_i = \mbox{Ber}(0)$ otherwise. We see that $Z_j = \prod_{i=1}^j (1-R_i)$ then simply indicates the event that $I$ has not been covered by the $j$th round. With these definitions, $\sum_{i=1}^{km/\epsilon} p_i Z_i = \sum_{i=1}^{km/\epsilon} \mathbb{E}[1_{I \in \cover_B[g] }]$ and
		\begin{align*}
			P\left(\sum\limits_{i=1}^{km/\epsilon} \mathbb{E}[1_{I \in \cover_B[g] }] > q\right) < \exp(-q).
		\end{align*}
		If $\sum\limits_{i=1}^{km/\epsilon} \mathbb{E}[1_{I \in \cover_B[g] }] < q$ then we say that the sequence $C$ is $q$-good. If a sequence is not $q$-good, it is called $q$-bad. If we let $q = \ln \delta_E^{-1}$, we see that the probability of an arbitrary sequence being $\ln \delta_E^{-1}$-good is at least $1-\delta_E$. This means that with probability $1-\delta_E$, 
		\begin{align*}
			\frac{P(\mathcal{A}(A)=C)}{P(\mathcal{A}(B)=C)}  &\leq \exp \left( \frac{e \epsilon_E \ln \delta_E^{-1} }{2} \right).
		\end{align*}
		Putting everything together, we see that this subroutine satisfies $\left( \frac{e \epsilon_E \ln \delta_E^{-1} }{2},\delta_E \right)$-differential privacy.
		
	\end{proof}
	
	\begin{proof}[Proof of \cref{thm:finalPrivacy}]
		We divide the privacy analysis into two halves; first, we bound the loss in privacy that occurs when constructing the proxy dataset $D''$. From \cref{lem:loopPrivacyLoss} we know that in the $m$ calls to \cref{alg:privateGridSetCover} the net loss in privacy is $(\frac{e \epsilon_E \log \delta_E^{-1}}{2}, \delta_E)$. In the calculation of noisy counts we see that two neighbouring datasets can only differ in their true counts by 1 unit at one center of $C$, from whence it follows that the $\ell_1$ sensitivity of the tuple of all counts is $1$ unit; this justifies the choice of parameter in the Laplace mechanism. Using basic composition \cref{thm:basicComp} along with the privacy loss bound for the Laplace mechanism \cref{lem:lapMech} we see that the net loss in privacy on releasing the proxy dataset $D''$ is $\left(\frac{e \epsilon_E \log \delta_E^{-1}}{2} + \epsilon_L, \delta_E \right)$. 
		
		We now have that $D''$ is publicly known and that the low-dimensional domain can be partitioned by identifying each point in the domain with the closest point in the set returned by the non-private clustering algorithm used (a Voronoi diagram). 
		
		In the second half of the analysis we use the parallel composition theorem (\cref{thm:parallelComp}) of \cite{mcsherry2009privacy} along with \cref{alg:noisyAVG} (\cref{thm:noisyAVG}). Since each application of \cref{alg:noisyAVG} on the separate clusters is $(\epsilon_G, \delta_G)$-differentially private, we apply parallel composition (\cref{thm:parallelComp}) to conclude that the net privacy loss over all $k$ applications is still $(\epsilon_G, \delta_G)$.
		
		Using basic composition we conclude that \cref{alg:privateKMeans} is $\left(\frac{e \epsilon_E \log \delta_E^{-1}}{2} + \epsilon_L + \epsilon_G, \delta_E + \delta_G \right)$-differentially private.
	\end{proof}

%% file: experiments.tex
\section{Experiments} \label{sec:experiments}

\begin{figure}[H]
	\centering
	\makebox[\textwidth]{%
		\begin{minipage}[t]{0.45\textwidth}
			\centering
			\begin{tikzpicture}
			\begin{axis}[
			title={Synthetic Dataset $\eps = 1$},
			xlabel={Centers},
			ylabel={K-Means Objective},
			xmin=0, xmax=20,
			ymin=6.9e8, ymax=1.055e9,
			xtick={2,6,10,14,18},
			ytick={7e8, 7.5e8, 8e8, 8.5e8, 9e8, 9.5e8, 1e9},
			legend style={at={(0.5,-0.2)},anchor=north},
			ymajorgrids=true,
			xmajorgrids=true,
			grid style=dashed,
			width=\textwidth
			]
			\addplot[purple] plot[ 
			error bars/.cd,
			y dir = both, y explicit
			]
			coordinates {
				(2,1.01446e9) +- (2.395755e6,2.395755e6)
				(6,9.87552e8) +- (2.5743e6,2.5743e6)
				(10,9.747e8) +- (1.639e6,1.639e6)
				(14,9.79e8) +- (6.1616e6,6.16161e6)
				(18,9.806712e8) +- (7.949e6,1.1815e6)
			};
			
			\addplot[blue] plot[ error bars/.cd,
			y dir = both, y explicit
			]
			coordinates {
				(2,1.00311e9) +- (2.914799e6,2.914799e6)
				(6,9.10188e8) +- (5.0831e6,5.0381e6)
				(10,8.25516e8) +- (1.1815e6,1.1815e6)
				(14,7.43920e8) +- (3.41029e6,3.41029e6)
				(18,7.0133e8) +- (5.771404e6,5.771404e6)
			};
			\addplot[green] plot[ 
			error bars/.cd,
			y dir = both, y explicit
			]
			coordinates {
				(2,1.0083e+09) +- (1.0738e+06,1.0738e+06)
				(6,1.0355e+09) +- (3.7897e+07,3.7897e+07)
				(10,1.0054e+09) +- (9.6128e+06,9.6128e+06)
				(14,9.9964e+08) +- (2.5999e+06,2.5999e+06)
				(18,1.0158e+09) +- (1.1279e+07,1.1279e+07)
			};
			\legend{Ours (\cref{alg:privateKMeans}), Non-private Lloyds, Balcan et al.}
			\end{axis}
			\end{tikzpicture}
			\label{fig:synth}
		\end{minipage} \hfill
		\begin{minipage}[t]{0.45\textwidth}
			\centering
			\begin{tikzpicture}
			\begin{axis}[
			title={MNIST Dataset $\eps = 1$},
			xlabel={Centers},
			ylabel={K-Means Objective},
			xmin=0, xmax=20,
			ymin=1.10e11, ymax=2.2e11,
			xtick={2,6,10,14,18},
			ytick={1.1e11, 1.25e11, 1.40e11, 1.55e11, 1.70e11, 1.85e11, 2e11, 2.15e11},
			legend style={at={(0.5,-0.2)},anchor=north},
			ymajorgrids=true,
			xmajorgrids=true,
			grid style=dashed,
			width=\textwidth
			]
			
			\addplot[purple] plot[ 
			error bars/.cd,
			y dir = both, y explicit
			]
			coordinates {
				(2,1.9685e11) +- (1.2536e9,1.2536e9)
				(6,1.8788e11) +- (1.141e9,1.141e9)
				(10,1.9007e11) +- (2.2242e9,2.2242e9)
				(14,1.9046e11) +- (1.8164e9,1.8164e9)
				(18,1.91789e11) +- (1.325e9,1.325e9)
			};
			\addplot[blue] plot[ error bars/.cd,
			y dir = both, y explicit
			]
			coordinates {
				(2,1.9435e11) +- (5.78e8,5.78e8)
				(6,1.77e11) +- (2.195e8,2.149e8)
				(10,1.58e11) +- (1.541e9,1.541e9)
				(14,1.43692e11) +- (2.322e8,2.322e8)
				(18,1.286e11) +- (2.584e8,2.584e8)
			};
			\addplot[green] plot[ 
			error bars/.cd,
			y dir = both, y explicit
			]
			coordinates {
				(2,2.0222e11) +- (2.3789e9,2.3789e9)
				(6,2.0765e11) +- (6.8912e9,6.8912e9)
				(10,1.9855e+11) +- (2.7748e9,2.7748e9)
				(14,1.9916e11) +- (4.7858e+09,4.7847e+09)
				(18,2.0476e+11) +- (5.7650e+09,5.7650e+09)
			};
			\legend{Ours (\cref{alg:privateKMeans}), Non-private Lloyds, Balcan et al.}
			\end{axis}
			\end{tikzpicture}
		\end{minipage}
	}%
	\caption{Empirical comparison of \cref{alg:privateKMeans} and the private $k$-means clustering algorithm from \cite{balcan2017differentially} }
	\label{fig:graphs}
\end{figure}
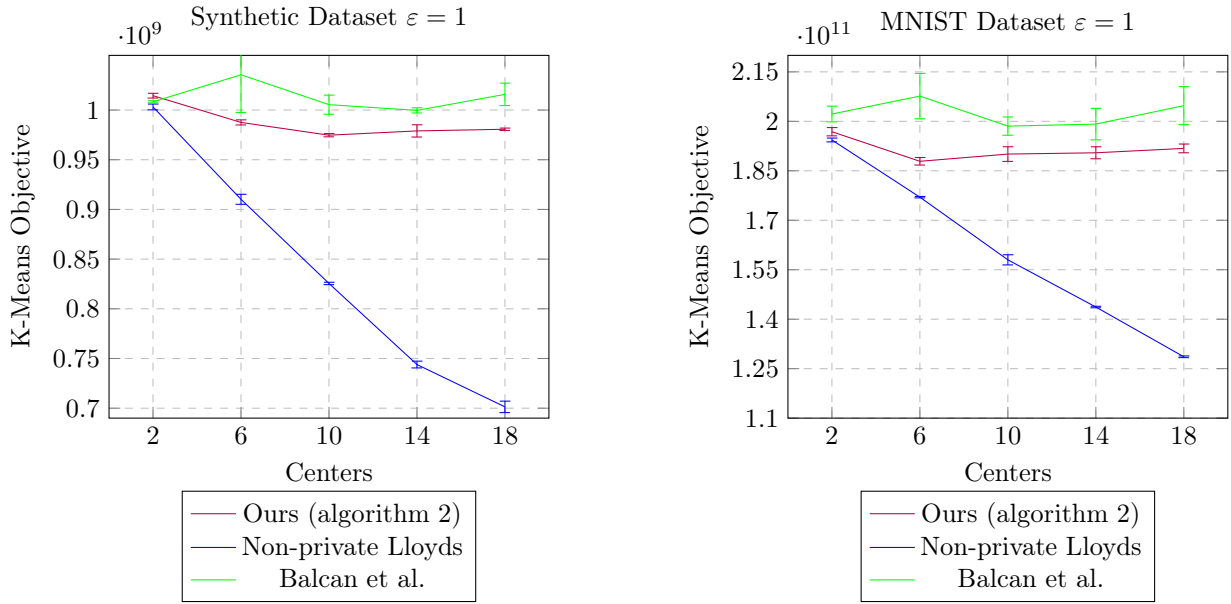
    
    In this section we present an experimental comparison between \cref{alg:privateKMeans}, the differentially private $k$-means clustering algorithm from \cite{balcan2017differentially}, and the non-private Lloyd's algorithm. Although there are other works with strong theoretical guarantees (such as \cite{KS18}), we are not aware of any implementation for those methods. The comparison here is done for two datasets; a synthetic dataset reproducing the construction in \cite{balcan2017differentially} and the MNIST dataset \cite{726791}. 
    
    The empirical results shown here for Balcan et al.'s algorithm~\cite{balcan2017differentially} come largely from their MATLAB implementation available on Github. Some corrections were made to the implementation of \cite{balcan2017differentially}; although the pseudocode uses a noisy count of the cluster sizes when computing the noisy average of the clusters found their implementation used the non-private exact count. We replaced this subroutine with \cref{alg:noisyAVG} to use the best method we know for privately computing the average.
	
	\paragraph{Implementation details:}{The privacy parameters were set to $\epsilon = 1$ and $\delta=n^{-1.5}$ for both algorithms. For each algorithm and dataset we let the number of centers $k = 2, 6, 10, 14$ and $18$. Our implementation of the algorithm, similar to \cite{balcan2017differentially}, projects to a smaller subspace of dimension size $\log(n)/2$ rather than $O(\log(n)/\epsilon^2)$ - note that this does not have any effect upon the privacy guarantee.}
		
	At the conclusion of both algorithms, we run one round of differentially private Lloyd's algorithm; adding this call to the differentially private Lloyd's yielded better empirical results for both the algorithm of \cite{balcan2017differentially} and ours. The addition of these rounds of Lloyd's requires adjusting privacy parameters by a constant factor but otherwise does not affect the privacy guarantees of the original algorithms. Although \cite{balcan2017differentially} satisfy $(\epsilon, 0)$ differential privacy and hence use the Laplace mechanism for their noisy average, we replaced this step with the noisyAVG routine of \cite{NSV16} for a fair comparison. The non-private Lloyd's algorithm was executed with 10 iterations. \Cref{fig:graphs} records the averages and standard deviation over five runs of each experiment.
	
	\paragraph{Datasets:}{The synthetic dataset is comprised of 50000 points randomly sampled from a mixture of 64 Gaussians in $\mathbb{R}^{100}$. The MNIST dataset uses the raw pixels; it is comprised of 70000 points with 784 features each.}
	
	\paragraph{Results:}{As can be seen in \cref{fig:graphs}, our algorithm achieves a lower $k$-means objective score than that of \cite{balcan2017differentially} for both the synthetic Dataset as well as the MNIST dataset. Similar to the experimental results in \cite{balcan2017differentially}, increasing the number of centers results in a decrease in the cost in the non-private algorithm but did not result in a concomitant decrease in the cost of the private algorithms. This behavior suggests that the algorithms are limited by their additive errors and that perhaps further decreasing them even in the constants would improve the gap compared with their non-private counterpart.}